\documentclass[a4paper,reqno]{amsart}
\usepackage{amsfonts}
\usepackage{tikz}
\usepackage{tkz-graph}
\usepackage{accents}
\usepackage{array,calc} 
\usepackage{stmaryrd} 
\usepackage{pgf}

\usepackage{textcase}

\usepackage{tikz}
\usepackage{tikz-cd}
\tikzcdset{column sep/normal=1.1cm}
\usetikzlibrary{calc}
\usepackage{color, mathtools,epsfig, graphicx}
\usetikzlibrary{positioning}
\usetikzlibrary{positioning, shapes, shadows, arrows}

\usepackage{amsmath,calc,graphicx}
\usepackage{url}
\usepackage[hidelinks]{hyperref}
\usepackage{amssymb}
\usepackage{amsthm}
\usepackage{multirow}
\usepackage{float}
\usepackage{enumerate}
\usepackage{enumitem}
\usepackage{amsrefs}
\usepackage{comment}
\usepackage{longtable}
\usepackage{todonotes}
\numberwithin{equation}{section}
\numberwithin{figure}{section}
\theoremstyle{plain}

\newtheorem{theorem}{Theorem}[section]
\newtheorem{lemma}[theorem]{Lemma}

\newtheorem{proposition}[theorem]{Proposition}
\newtheorem{definition}[theorem]{Definition}

\newtheorem{conjecture}[theorem]{Conjecture}

\theoremstyle{remark}
\newtheorem{remark}[theorem]{Remark}

\newtheorem*{lem*}{\textsc{Lemma}}
\newtheorem*{cor*}{\textsc{Corollary}}
\newtheorem*{exer*}{\textsc{Exercise}}
\newtheorem*{con*}{\textsc{Conjecture}}
\newtheorem*{thm*}{\textsc{Theorem}}

\newcommand{\beq}{\begin{equation}}
\newcommand{\eeq}{\end{equation}}

\makeatletter
\newcommand*\rel@kern[1]{\kern#1\dimexpr\macc@kerna}
\newcommand*\widebar[1]{%
  \begingroup
  \def\mathaccent##1##2{%
    \rel@kern{0.8}%
    \overline{\rel@kern{-0.8}\macc@nucleus\rel@kern{0.2}}%
    \rel@kern{-0.2}%
  }%
  \macc@depth\@ne
  \let\math@bgroup\@empty \let\math@egroup\macc@set@skewchar
  \mathsurround\z@ \frozen@everymath{\mathgroup\macc@group\relax}%
  \macc@set@skewchar\relax
  \let\mathaccentV\macc@nested@a
  \macc@nested@a\relax111{#1}%
  \endgroup
}
\makeatother

\newcommand{\change}[1]{#1}

\newcommand\Pone{\textrm{P}_{\textrm{I}} }
\newcommand\Ptwo{\textrm{P}_{\textrm{II}} }

\newcommand\Pfour{\textrm{P}_{\textrm{IV}}}

\newcommand{\orcidauthorA}{0000-0001-7504-4444}
\newcommand{\orcidauthorB}{0000-0002-0461-7580}

\title[On integrals in finite characteristic]{On integrals of non-autonomous dynamical systems in finite characteristic}
\date{}
\thanks{This research was supported by NISDRG Grant \#NI240100145.}
\author[Nalini Joshi]{Nalini Joshi}
\thanks{NJ's ORCID ID is \orcidauthorA.}
\address{School of Mathematics and Statistics F07, The University of Sydney, NSW 2006, Australia}
\author{Pieter Roffelsen}
\thanks{PR's ORCID ID is \orcidauthorB.}
\email{nalini.joshi@sydney.edu.au}
\email{pieter.roffelsen@sydney.edu.au}

\begin{document}
\begin{abstract}
We use a difference Lax form to construct simultaneous integrals of motion of the fourth Painlev\'e equation and the difference second Painlev\'e equation over fields with finite characteristic $p>0$. For $p\neq 3$, we show that the integrals can be normalised to be completely invariant under the corresponding extended affine Weyl group action. We show that components of reducible fibres of integrals correspond to reductions to Riccati equations. We further describe a method to construct non-rational algebraic solutions in a given positive characteristic. We also discuss a projective reduction of the integrals.
\end{abstract}
\maketitle

\section{Introduction}

Given a field $\mathbf{k}$, consider the field of rational functions $\mathbf{k}(f,g,\alpha_1,\alpha_2,t)$ with the $\mathbf{k}$-linear derivation $D_t$ specified by
  \begin{equation}\label{eq:piv}
 \Pfour:\quad           \begin{cases}
        D_t f=f(2\hspace{0.3mm}g-f-t)-\alpha_1, &   D_t \alpha_1=0, \\
   D_t \hspace{0.3mm}g=g\hspace{0.3mm}(2f-\hspace{0.3mm}g+t)+\alpha_2, &   D_t \alpha_2=0, \\
        D_t \hspace{0.4mm}t\hspace{0.4mm}=1. & 
    \end{cases}
    \end{equation}
 This system of ODEs is also known as the fourth Painlev\'e equation  \cite{okamoto86}\footnote{\change{See Section \ref{subsec:notation} on how our normalisations of $\Pfour$ and $d\Ptwo$ compare with those in Okamoto \cite{okamoto86} and Noumi and Yamada \cite{noumiyamadadiscrete} respectively.}
 }. 
The derivation $D_t$ commutes with the \change{two} $\mathbf{k}$-linear field isomorphisms specified by \change{their actions on the generators as}
\begin{align*}
   && T_1:& &(f,g,\alpha_1,\alpha_2,t)\mapsto (\overline{f}^1,\overline{g}^1,\alpha_1+1,\alpha_2,t), &&\\
  &&  T_2:& &(f,g,\alpha_1,\alpha_2,t)\mapsto (\overline{f}^2,\overline{g}^2,\alpha_1,\alpha_2+1,t), &&
\end{align*}
where
\begin{equation}\label{eq:painleve_equation}
\textrm{d}\Ptwo:\quad\begin{cases}
\begin{aligned}
    \overline{f}^1&=\frac{(\alpha_1+\alpha_2)f}{\alpha_1-fg}-f+g-t-\frac{\alpha_1}{f}, & 
     \overline{g}^1&=\frac{(\alpha_2+fg)f}{\alpha_1-fg},\\
     \overline{g}^2&=\frac{(\alpha_1+\alpha_2)g}{\alpha_2+fg}+f-g+t+\frac{\alpha_2}{g}, &     \overline{f}^2&=\frac{(\alpha_1-fg)g}{\alpha_2+fg}.
\end{aligned}&
\end{cases}
\end{equation}
This set of equations defines a discrete dynamical system known as the  \change{difference (or discrete) second Painlev\'e equation \cite[\S 5]{noumiyamadadiscrete}$^1$}. We refer to \eqref{eq:piv} and \eqref{eq:painleve_equation} as $\Pfour$ and 
 $\textrm{d}\Ptwo$ over $\mathbf{k}$ respectively.

In this paper, we study $\Pfour$ and $\textrm{d}\Ptwo$ over fields $\mathbf{k}$ with finite characteristic $p>0$. Whereas, over the complex numbers, Painlev\'e equations only have transcendental integrals of motion \cites{gromaknointegrals,gromakbook,incebook,malgrange1,malgrange2,nishioka1988,umemura1997},  e.g. constructible through the Riemann-Hilbert correspondence \cites{chekhov2017,itsbook,painleverhbook,putsaito,JR2021,JR2023,R2024}, we show that the situation is drastically different in finite characteristic. We focus on $\Pfour$ and \change{$\textrm{d}\Ptwo$}, though our methodology extends to other Painlev\'e equations that have a difference Lax pair, and show that they admit simultaneous polynomial integrals of motion in finite characteristic. We provide a method to construct these integrals of motion explicitly based on a difference Lax form by Nakazono \cite{nakazono_lax} for $\Pfour$.

These results \change{parallel} earlier results on a $q$-difference Painlev\'e equation obtained in \cite{JR2026}, where integrals of motion are constructed for $q\Pone$ using a $q$-difference Lax form when $q$ has finite multiplicative order. Computational experiments in that paper suggest that fibres of the integrals of motion are generally elliptic curves and indeed it was checked for more than 200M orbits over finite fields that they satisfy an analogue of the Hasse bound.

In the current paper, we do not carry out a computational study of orbits and fibres, but instead relate geometric properties of fibres to aspects of solutions of $\Pfour$ and \change{$\textrm{d}\Ptwo$}. We show how components of reducible fibres of integrals of motion are related to reductions of $\Pfour$ to Riccati equations. We also construct non-rational algebraic solutions of $\Pfour$ and \change{$\textrm{d}\Ptwo$}, which correspond to singularities of fibres. The latter solutions have no apparent reflection in characteristic zero.

We further show that the integrals of motion reduce to integrals of the difference first Painlev\'e  equation ($\change{\textrm{d}\Pone}$) under a projective reduction.

Whilst our results apply whenever the field $\mathbf{k}$ has finite characteristic $p$, their most interesting applications are likely to lie in the setting of finite fields, $\mathbf{k}=\mathbb{F}_q$, $q=p^n$ where $p$ prime and $n\geq 1$ integer. In this context, they provide a stepping stone toward the largely unexplored arithmetic study of Painlev\'e equations\change{, which is one of our main motivations for studying Painlev\'e equations in finite characteristic.} 

\change{The arithmetic study of linear ODEs has a rich history, with deep and long-standing open problems like the Grothendieck–Katz $p$-curvature conjecture \cite{katz72}. Of particular interest is the possibility of using reductions modulo primes to study linear differential equations in characteristic zero, for example to investigate their differential Galois and monodromy groups \cite{katz82} or to construct algebraic and exponential solutions \cite{bostancaruso24,cluzeauhoeij04}.}

\change{Recently, results of this flavour have also been obtained for nonlinear ODEs \cite{lamlitt25,lamlitt26}, including ones coming from Fuchsian isomonodromic systems, like Painlevé VI, with initial conditions corresponding to Picard–Fuchs equations. These developments suggest that the arithmetic study of Painlev\'e equations may reveal new properties that are not readily apparent from a purely complex-analytic perspective.}

\change{A further motivation for studying Painlev\'e equations in positive characteristic is the possibility of cryptographic applications. Painlev\'e transcendents generalise elliptic functions in a variety of ways, and the arithmetic theory of elliptic curves underpins many widely used cryptographic protocols. This raises the question whether the arithmetic theory of Painlev\'e equations might likewise find applications in cryptography.}

\subsection{Comparison of normalisation and notation with the literature}\label{subsec:notation}\change{
We note that our normalisation of $\Pfour$ differs slightly from that in Okamoto \cite[Eq. (0.1)]{okamoto86} and, in characteristic different from $2$, after adjoining a $c$ with $c^2=2$ if necessary, the two are related by
\begin{equation*}
    q^{\operatorname{OK}}(t)=c\,f(c\,t),\quad p^{\operatorname{OK}}(t)=c^{-1}g(c\,t),\quad \kappa_0^{\operatorname{OK}}=\alpha_1,\quad \theta_\infty^{\operatorname{OK}}=-\alpha_2,
\end{equation*}
 where $q=q^{\operatorname{OK}}$ satisfies the following standard scalar form of $\Pfour$,
\begin{align*}
    D_t^2q=&\frac{(D_tq)^2}{2q}+\frac{3}{2}q^3+4tq^2+2(t^2-\alpha)q+\frac{\beta}{q},
\end{align*}
with
\begin{equation*}
    \alpha=\alpha_0-\alpha_2,\quad \beta=-2\alpha_1^2,
\end{equation*}
where $\alpha_0$ is defined by the equation $\alpha_0+\alpha_1+\alpha_2=1$.
}

\change{
In the literature, the notation $\textrm{d}\Ptwo$ is often used to refer to dynamics induced by a single translation direction. We use it to refer to the dynamical system induced by both $T_1$ and $T_2$. For $j=1,2$ and $h\in \mathbf{k}(f,g,\alpha_1,\alpha_2,t)$, we denote by $\overline{h}^{j}$ the result of applying translation $T_j$ to $h$. More generally, if we have an object in a matrix polynomial ring, power series ring or similar over $\mathbf{k}(f,g,\alpha_1,\alpha_2,t)$, then application of $\overline{\cdot}^{j}$ is defined component and coefficient-wise, replacing any instance of $f,g,\alpha_1,\alpha_2,t$ by its image under $T_j$.
}

\change{Our normalisation of $\Pfour$ relates to the symmetric form of $\Pfour$ in Noumi and Yamada \cite[\S 5]{noumiyamadadiscrete} via the change of variables
\begin{align*}
f_0^{\operatorname{NY}}&=t+f-g, & f_1^{\operatorname{NY}}&=-f, & f_2^{\operatorname{NY}}&=g, & \delta^{\operatorname{NY}}&=1,\\
\alpha_0^{\operatorname{NY}}&=\alpha_0, & 
 \alpha_1^{\operatorname{NY}}&=\alpha_1, &  \alpha_2^{\operatorname{NY}}&=\alpha_2,  &  c^{\operatorname{NY}}&=t.
\end{align*}
Then, the specific dynamical system \cite[Eq. (5.7)] {noumiyamadadiscrete} is induced by the translation $T^{\operatorname{NY}}=T_1^{-1}$. Using symmetric variables, the translations take slightly easier forms with respect to some of the variables, for example}
\begin{equation*}
 \change{T_1(f_0^{\operatorname{NY}})+f_0^{\operatorname{NY}}=t-f_1^{\operatorname{NY}}+\frac{\alpha_1}{f_1^{\operatorname{NY}}}.}
\end{equation*}

\change{
We remark that historically, the first discrete Painlev\'e equation associated with $\Pfour$ in the literature is 
discrete Painlev\'e one ($\textrm{d}\Pone$)  \cite{fokasitskitaev91}. However, $\textrm{d}\Pone$ is only associated with $\Pfour$ under special parameter constraints, as the corresponding element of the symmetry group is generally a quasi-translation, see Section \ref{sec:projred} for more details. The discrete Painlev\'e equation associated with $\Pfour$ in general is $\textrm{d}\Ptwo$ and governs Schlesinger transformations of an associated isomonodromic system, as implicitly described in \cite[Prop. 3.4]{fokasmugan88}.
}

\subsection{Overview paper}
Section \ref{sec:mainresults} describes the main results of the paper. In Section \ref{sec:integralmotanalysis}, the integrals of motion are analysed and the action of the affine Weyl symmetry group on them is computed. Section \ref{sec:fibres} is concerned with fibres of the integrals of motion and special solutions of \change{$\textrm{d}\Ptwo$} and $\Pfour$. Projective reductions of the integrals of motion are discussed in Section \ref{sec:projred}. A conclusion is provided in Section \ref{sec:conclusion}, followed by three appendices that contain explicit formulas and proofs of some technical lemmas.

\subsection{Acknowledgements}
This research was funded by the Australian Government through the Office of National Intelligence grant \# NI240100145. We thank Tomas Lasic Latimer for fruitful discussions regarding trace identities in finite characteristic that led to the proof of Lemma \ref{lem:traceconstant}. We also like to thank Claude Viallet and John Roberts for interesting discussions related to the topic of the paper.

\section{Main Results}\label{sec:mainresults}

Recall that, when a field $\mathbf{k}$ has finite characteristic $p>0$, there is a canonical embedding $\mathbb{F}_p\hookrightarrow \mathbf{k}$ and we will consider $\mathbb{F}_p\subseteq \mathbf{k}$ as a subfield through this embedding.
Our first main result is the existence of a simultaneous integral of motion $\mathcal{I}_p$ of $\textrm{d}\Ptwo$ and $\Pfour$, over fields with characteristic $p>0$,  that is polynomial in all the relevant variables.

The construction of the integral $\mathcal{I}_p$ is based on the difference Lax form given by Nakazono \cite{nakazono_lax}. Of particular relevance is the following $2\times 2$ matrix polynomial \change{with respect to} $z$,
\begin{gather}
\begin{aligned}
    A(z)=
    \begin{bmatrix}
        w & 0\\
        0 & 1
    \end{bmatrix}
    \begin{bmatrix}
        1 & 0\\
        f-t\, f\,g-t\,z & 1
    \end{bmatrix}
    \begin{bmatrix}
        g & 1\\
        z+\alpha_2 & 0
    \end{bmatrix}\hspace{2cm}\\
    \hspace{2cm}\cdot
    \begin{bmatrix}
        -f & 1\\
        z & 0
    \end{bmatrix}
    \begin{bmatrix}
        t\,z-g+t\,f\,g & 1\\
        z-\alpha_1 & 0
    \end{bmatrix}
    \begin{bmatrix}
        w^{-1} & 0\\
        0 & 1
    \end{bmatrix},
    \end{aligned}
\label{eq:matrixAdef}
\end{gather}
with coefficients in $\mathbb{Z}[f,g,\alpha_1,\alpha_2,t,w,w^{-1}]$\change{, see \eqref{eq:coefficient_matrix} for the explicit matrix coefficients after expanding}. This is the coefficient matrix of the spectral equation of the Lax form in \cite{nakazono_lax}, after some rescaling. It contains a free scaling parameter $w$ whose dependence falls out in the construction of integrals of motion.

\begin{definition}\label{def:integralofmotion}
    For any prime $p\geq 2$, define the polynomial $\mathcal{I}_p\in \mathbb{F}_p[f,g,\alpha_1,\alpha_2,t]$ by
\begin{equation*}
    \mathcal{I}_p:=\operatorname{Tr}[A(p-1)\cdot A(p-2)\cdot\ldots\cdot A(1)\cdot A(0)],
\end{equation*}
in characteristic $p$, where $A(z)$ is the matrix polynomial defined in equation \eqref{eq:matrixAdef}.
\end{definition}

Our first main result is the following theorem, proven in Section \ref{sec:lax}.
\begin{theorem}\label{thm:main_theorem}
For any prime $p\geq 2$, the polynomial $\mathcal{I}_p$ is a simultaneous integral of motion for $\textrm{d}\Ptwo$ and $\Pfour$ over fields of characteristic $p$, that is,
\begin{equation*}
    \overline{\mathcal{I}}_p^1=\overline{\mathcal{I}}_p^2=\mathcal{I}_p,\qquad D_t \mathcal{I}_p=0,
\end{equation*}
of total degree $3p$ and bi-degree $(2p,2p)$ in $(f,g)$ with
\begin{equation*}
\mathcal{I}_p=-f^{2p}g^p+f^{p}g^{2p}+O(f,g),
\end{equation*}
for a polynomial $O(f,g)\in \mathbb{F}_p[f,g,\alpha_1,\alpha_2,t]$ of degree less than $2p$ in $f$ and degree less than $2p$ in $g$. 
\end{theorem}
In Appendix \ref{appendix:explicitformulas}, explicit formulas for the first couple of integrals of motion are given.

To describe our second main result, we recall that there is a $\mathbf{k}$-linear action of the extended affine Weyl group 
$\widetilde W(A_2^{(1)})$ on     $\mathbf{k}(f,g,\alpha_1,\alpha_2,t)$ that commutes with the derivation $D_t$ \cite{noumibook,noumiyamada99}. Realising $\widetilde W(A_2^{(1)})$ through the generators $\{s_0,s_1,s_2,\pi\}$ with relations
\begin{equation*}
  s_k^2=1,\quad \pi^3=1,\quad  \change{ \pi s_{k}=s_{k+1}\pi,}\quad (s_k s_{k+1})^3=1,
\end{equation*}
for $k\in\mathbb{Z}/{3\mathbb{Z}}$,  this action is defined in Table \ref{table:symmetryactions}. The discrete dynamics in \eqref{eq:painleve_equation} then correspond to the translational elements
\begin{equation*}
 \change{   T_1=s_1s_2\pi^2,\qquad
    T_2=s_2s_1\pi.}
\end{equation*}

\begin{table}
\caption{Actions of the generators of $\widetilde W(A_2^{(1)})$ on $t,\alpha_1,\alpha_2,f,g$.}\label{table:symmetryactions}
\begin{tabular}{|c||c|cc|cc|}
\hline
&$t$&$\alpha_1$&$\alpha_2$&$f$&$g$\\
\hline
\rule{0pt}{4ex}$s_0$&$t$&$1-\alpha_2$&$1-\alpha_1$&$\displaystyle f+\frac{\alpha_1+\alpha_2-1}{f-g+t}$&$\displaystyle g+\frac{\alpha_1+\alpha_2-1}{f-g+t}$\\[2ex]\hline
\rule{0pt}{4ex}$s_1$&$t$&$-\alpha_1$&$\alpha_2+\alpha_1$&$f$&$\displaystyle g-\frac{\alpha_1}{f}$\\[2ex]\hline
\rule{0pt}{4ex}$s_2$&$t$&$\alpha_1+\alpha_2$&$-\alpha_2$&$\displaystyle f+\frac{\alpha_2}{g}$&$g$\\[2ex]
\hline
\rule{0pt}{4ex}$\pi$&$t$&$\alpha_2$&$1-\alpha_1-\alpha_2$&$-g$&$f-g+t$\\[2ex]\hline
\end{tabular}
\end{table}

 For each integral of motion $\mathcal{I}_p$, for primes $p\neq 3$, we define an associated integral of motion
\begin{equation}\label{eq:normalisation}
    \widetilde{\mathcal{I}}_p=\begin{cases}
        \mathcal{I}_p+\tfrac{1}{3}t^p(\alpha_1^p-\alpha_1-\alpha_2^p+\alpha_2), & \text{if $p\geq 5$,}\\
        \mathcal{I}_2+(1+t^2)(\alpha_1^2+\alpha_1+\alpha_2^2+\alpha_2) & \text{if $p=2$.}
    \end{cases}
\end{equation}
Then our second main result is the following theorem, proven in Section \ref{sec:sym}.
\begin{theorem}\label{thm:invariance}
    For $p\neq 3$, the associated integral of motion $\widetilde{\mathcal{I}}_p$ is completely invariant under the extended affine Weyl group action, that is,
\begin{equation*}
   w(\widetilde{\mathcal{I}}_p)=\widetilde{\mathcal{I}}_p\qquad (w\in \widetilde W(A_2^{(1)})).
\end{equation*}
\end{theorem}
\begin{remark}
    The action of the extended affine Weyl group action on $\mathcal{I}_3$ is given in Theorem \ref{thm:sym}. Whether it can be normalised so that it is invariant under the extended affine Weyl group action, like the others, remains to be investigated.
\end{remark}

The equations \change{$\textrm{d}\Ptwo$} and $\Pfour$ can be simultaneously regularised on a rational surface known as the initial value space \cites{okamoto1979,s:01}.
Our third main result is Theorem \ref{thm:morphism}, which shows that the integrals of motions define regular maps on the initial value space of \change{$\textrm{d}\Ptwo$} and $\Pfour$. 

In Section \ref{sec:fibre_reducible}, we explain how reductions of $\Pfour$ to Riccati equations for special parameter values lead to reducible fibres of the integrals of motion. We further conjecture that this is the only way reducible fibres may arise, see Conjecture \ref{conj:irred}. In Section \ref{sec:ratsol} we discuss the particular case of rational solutions and in Section \ref{subsec:singularities} we show that singularities on fibres correspond to non-rational algebraic solutions of \change{$\textrm{d}\Ptwo$} and $\Pfour$.

The last main result of the paper is that the integrals reduce to integrals of motion of \change{$\textrm{d}\Pone$} under a projective reduction from \change{$\textrm{d}\Ptwo$} to \change{$\textrm{d}\Pone$}, as proven in Section \ref{sec:projred}.

\section{Analysis of integrals of motion}\label{sec:integralmotanalysis}




In this section we show that the polynomials $\mathcal{I}_p$, for prime $p$, as defined in Definition \ref{def:integralofmotion}, are integrals of motion of \change{$\textrm{d}\Ptwo$} and $\Pfour$ and work out how the affine Weyl symmetry group acts on them, in particular
proving Theorems \ref{thm:main_theorem} and \ref{thm:invariance}.

\subsection{Lax form}\label{sec:lax}
The construction of the integrals of motions is based on a
coupled set of systems of difference and differential equations that form Lax pairs for $\Pfour$ and \change{$\textrm{d}\Ptwo$}, where, crucially, the spectral equation is a difference equation. We took this Lax form from \cite{nakazono_lax}. It reads as follows after some rescaling,
\begin{subequations}\label{eq:piv_lax}
    \begin{align}
        Y(z+1)&=A(z)Y(z),\label{eq:piv_laxA}\\\
        \overline{Y}^j(z)&=B^{(j)}(z)Y(z)\qquad (j=1,2), \label{eq:piv_laxB}\\
        D_t Y(z)&=C(z)Y(z) \label{eq:piv_laxC}
    \end{align}
    \end{subequations}
where $A(z)$ is the $2\times 2$ matrix polynomial defined in \eqref{eq:matrixAdef}. Direct computation yields that it is a matrix polynomial of degree three,
\begin{equation}
\label{eq:coefficient_matrix}
    A(z)=A_0+z A_1+z^2 A_2+z^3 A_3,   
\end{equation}
with coefficients given by
\begin{equation*}
    A_3=\begin{bmatrix}
        0 & 0\\
        -w^{-1}t^2 & 0
    \end{bmatrix}, \quad
    A_2=\begin{bmatrix}
        t & 0\\
        w^{-1}(1-t^2fg) & -t
    \end{bmatrix},\quad
    A_1=\begin{bmatrix}
        0 & w\\
        w^{-1}h & 0
    \end{bmatrix},
\end{equation*}
and
\begin{equation*}
        A_0=\begin{bmatrix}
        -\alpha_1g+fg^2-t f^2g^2 & -w fg\\
        w^{-1}(fg\,h-\alpha_1\alpha_2) & -\alpha_2f-f^2g+t f^2g^2
    \end{bmatrix}, 
\end{equation*}
where
\begin{equation*}
    h:=fg(t\,g-1)(t f-1)+\alpha_1(t\,g-1)-\alpha_2(t f-1).
\end{equation*}
Explicit expressions for the coefficient matrices $B^{(1)}(z)$, $B^{(2)}(z)$ and $C(z)$ are given in Appendix \ref{sec:laxcoef}, but they will not be used in this paper.

For \change{fixed $t$, $\alpha_1$ and $\alpha_2$}, the coefficient matrix $A(z)$ is characterised by the following two analytic properties.
\begin{enumerate}[label=(\roman*)]
    \item $A(z)$ is a $2\times 2$ matrix polynomial with
   \begin{equation}\label{eq:coefmatrixnormalisation}
    \begin{aligned}
        A_{11}(z)&=t\,z^2+\mathcal{O}(1), &
        A_{12}(z)&=\mathcal{O}(z),\\
        A_{21}(z)&=\mathcal{O}(z^3), &
        A_{22}(z)&=-t\,z^2+\mathcal{O}(1),
    \end{aligned}
   \end{equation}
    as $z\rightarrow \infty$.
    \item The determinant of $A(z)$ equals
\begin{equation}\label{eq:coefmatrixdet}
    |A(z)|=-z(z-\alpha_1)(z+\alpha_2).
\end{equation}
\end{enumerate}

For $j=1,2$, the transformation
\begin{equation}\label{eq:translation_action_coef}
    A(z)\mapsto \overline{A}^j(z)=B^{(j)}(z+1)A(z)B^{(j)}(z)^{-1}
\end{equation}
shifts one of the zeros of the determinant, namely $\alpha_j\mapsto\alpha_j+1$, and the resulting evolution of $(f,g)$ is given precisely by the translation $T_j$, with $\overline{w}^{j}=w$.

Similarly, compatibility of \eqref{eq:piv_laxA} and \eqref{eq:piv_laxC} yields
\begin{equation}\label{eq:coefdif}
    D_t A(z)=C(z+1)A(z)-A(z)C(z),
\end{equation}
which is equivalent to \eqref{eq:piv} and $D_t w=0$.

In \cite{nakazono_lax} only the deformation equation \eqref{eq:piv_laxC} is given. The Schlesinger transformations \eqref{eq:piv_laxB} were obtained by direct computation, i.e. by solving, see equation \eqref{eq:translation_action_coef},
\begin{equation*}
    \overline{A}^j(z)B^{(j)}(z)=B^{(j)}(z+1)A(z),
\end{equation*}
for $j=1,2$.

\begin{remark}\label{rem:lax equivalent}
There is also a scalar Lax form available in the literature \cite{KNY2017}*{8.5.20(ii)}. Denoting by $y(z)$ the dependent variable of the spectral equation of this Lax form, it can be related
to $Y(z)$ in equation \eqref{eq:piv_laxA} by
    \begin{equation*}
        Y(z)=G(z)\begin{bmatrix}
            \psi(z)\\
            \psi(z+1)
        \end{bmatrix},\quad \psi(z):=\Gamma(z)\Gamma(z-\alpha_1) y(z),
    \end{equation*}
    where
    \begin{equation*}
        G(z)=\begin{bmatrix}
            1 & 0\\
            \frac{fg^2(tf-1)+\alpha_1 g-tz^2}{w(z-f g)} & \frac{1}{w(z-f g)}
        \end{bmatrix},
    \end{equation*}
    where we note  the correspondence  $(f^{\text{KNY}},g^{\text{KNY}})=(f,fg)$ between the $(f,g)$ coordinates used in \cite{KNY2017} and ours. Matching up the deformation equations \eqref{eq:piv_laxB} and \eqref{eq:piv_laxC} with those in \cite{KNY2017} may require an additional $z$-independent scaling.
\end{remark}

\subsection{Proof of Theorem \ref{thm:main_theorem}}
Fix a prime $p\geq 2$. To prove theorem \ref{thm:main_theorem}, it will be useful to consider the following matrix polynomial,
\begin{equation}\label{eq:Mpdefi}
    M_p(z)=A(z+p-1)\cdot A(z+p-2)\cdot\ldots\cdot A(z+1)\cdot A(z),
\end{equation}
in characteristic $p$.
Note that this matrix polynomial satisfies
\begin{equation*}
    M_p(z+1)=A(z)M_p(z)A(z)^{-1}.
\end{equation*}
Furthermore, by equations \eqref{eq:translation_action_coef} and \eqref{eq:coefdif},
\begin{align}
    \overline{M}_p^j(z)&=B^{(j)}(z)M_p(z)B^{(j)}(z)^{-1},\\
    D_t M_p(z)&=C(z)M_p(z)-M_p(z)C(z).
\end{align}
This means that the trace
\begin{equation}\label{eq:defi_chi}
    \chi_p(z):=\operatorname{Tr}[M_p(z)], 
\end{equation}
is a polynomial in $z$, with coefficients in $\mathbb{F}_p[f,g,\alpha_1,\alpha_2,t]$,  that satisfies
\begin{subequations}
\begin{align}
    \chi_p(z+1)&=\chi_p(z),\\
    \overline{\chi}_p^j(z)&=\chi_p(z),\label{eq:shift}\\
    D_t \chi_p(z)&=0.\label{eq:Dt}
\end{align}
\end{subequations}
In other words, its coefficients are simultaneously integrals of motion of $\Pfour$ and \change{$\textrm{d}\Ptwo$}. There are only three relevant coefficients, as the following lemma shows.
\begin{lemma}\label{lem:chi_form}
The polynomial $ \chi_p(z)$ takes the form
\begin{equation*}
    \chi_p(z)=i_0+i_1 z+i_p z^p,
\end{equation*}
for some $i_0,i_1,i_p\in \mathbb{F}_p[f,g,\alpha_1,\alpha_2,t]$, with 
\begin{equation*}
    \change{i_0=\mathcal{I}_p},\qquad i_1+i_p=0,
\end{equation*}
where $\change{\mathcal{I}_p}$ is defined in Definition \ref{def:integralofmotion}.
\end{lemma}
\begin{proof}
To prove the lemma, we estimate the degree of $\chi_p(z)$ with respect to $z$.

For any $n\geq 1$, denote
\begin{equation*}
    M_n(z)=A(z+n-1)\cdot A(z+n-2)\cdot\ldots\cdot A(z+1)\cdot A(z).
\end{equation*}
Using
\begin{equation*}
    A(z)=\begin{bmatrix}
        t z^2+\mathcal{O}(1) & wz+\mathcal{O}(1)\\
        -w^{-1}t^2 z^3+\mathcal{O}(z^2) & -t z^2+\mathcal{O}(1)
    \end{bmatrix}\qquad (z\rightarrow \infty),
\end{equation*}
and induction on $n$, we obtain
\begin{equation*}
    M_n(z)=\begin{bmatrix}
        \mathcal{O}(z^{2n-1}) & \mathcal{O}(z^{2n-2})\\
        \mathcal{O}(z^{2n}) & \mathcal{O}(z^{2n-1})
    \end{bmatrix}\qquad (z\rightarrow \infty),
\end{equation*}
for $n\geq 2$.
It follows that $\chi_p(z)=\operatorname{Tr}M_p(z)$ has degree at most $2p-1$ in $z$.

Using only this fact and the identity $\chi_p(z+1)=\chi_p(z)$, it follows that $\chi_p(z)$ must take the form
\begin{equation*}
    \chi_p(z)=i_0+i_1 z+i_p z^p,
\end{equation*}
for some $i_0,i_1,i_p\in \mathbb{F}_p[f,g,\alpha_1,\alpha_2,t]$, with $i_1+i_p=0$. Finally, the identity \change{$i_0=\mathcal{I}_p$} follows directly from Definition \ref{def:integralofmotion} and the lemma follows.
\end{proof}

We now prove Theorem \ref{thm:main_theorem}.
\begin{proof}[Proof of Theorem \ref{thm:main_theorem}]
By equations \eqref{eq:shift} and \eqref{eq:Dt}, all the coefficients of the polynomial $\chi_p(z)$ are integrals of motion of \change{$\textrm{d}\Ptwo$} and $\Pfour$. Since, in the notation of Lemma \ref{lem:chi_form}, 
\begin{equation*}
    \mathcal{I}_p=\chi_p(0)=i_0\in \mathbb{F}_p[f,g,\alpha_1,\alpha_2,t],
\end{equation*}
it follows that $\mathcal{I}_p$ is a polynomial simultaneous integral of motion of \change{$\textrm{d}\Ptwo$} and $\Pfour$.

To prove the theorem, it remains to estimate the degrees of $\mathcal{I}_p$ in $f$ and $g$. We first consider the degree with respect to $f$.
By rescaling $w=\widehat{w}f$, which does not affect the trace $\chi_p(z)$, we have that the coefficient matrix satisfies
\begin{equation*}
    A(z)=f^2 g\, U+\mathcal{O}(f),\qquad U=\begin{bmatrix}
        -g t & -\widehat{w}\\
       \widehat{w}^{-1} gt(gt-1) & gt-1
    \end{bmatrix},
\end{equation*}
as $f\rightarrow \infty$. Now, the matrix $U$ satisfies the identity $U^2=-U$, so that
\begin{align*}
    \chi_p(z)=& \operatorname{Tr}\left[(f^2 g)^p U^p+\mathcal{O}(f^{2p-1})\right]\\
    =&(-1)^{p+1}(f^2 g)^p \operatorname{Tr} U+\mathcal{O}(f^{2p-1})\\
    =&-(f^2g)^p+\mathcal{O}(f^{2p-1}),
\end{align*}
 and consequently
\begin{equation*}
    \mathcal{I}_p=\chi_p(0)=-(f^2g)^p+\mathcal{O}(f^{2p-1}),
\end{equation*}
as $f\rightarrow \infty$. Completely analogously,  we derive the highest degree term in $g$.

What is left, is to show that the total degree of $\mathcal{I}_p$ in $(f,g)$ is $3p$. We already know that the total degree is at least $3p$. For the case $p=2$, the assertion follows from the explicit formula for $\mathcal{I}_2$ in Appendix \ref{appendix:explicitformulas}. So, from here on, we assume that $p$ is an odd prime.

To estimate the total degree from above, we rescale $w=\widehat{w}fg$ and write
\begin{equation*}
    f=c_1\lambda,\quad g=c_2\lambda.
\end{equation*}
Substitution into the coefficient matrix then gives
\begin{equation*}
    A(z)=\lambda^4(U_0+\lambda^{-1}U_1+\lambda^{-2} U_2(z)+\lambda^{-3} U_3(z)+\ldots),
\end{equation*}
where
\begin{align*}
    U_0&=c_1^2 c_2^2\begin{bmatrix}
        -t & -\widehat{w}\\
      \widehat{w}^{-1}t^2 & t  
    \end{bmatrix},\quad
    U_1=c_1 c_2\begin{bmatrix}
        c_2 & 0\\
      -\widehat{w}^{-1}(c_1+c_2)t & -c_1  
    \end{bmatrix},\\
    U_2(z)&=c_1c_2\begin{bmatrix}
        0 & \widehat{w}z\\
        \widehat{w}^{-1}(1+t^2z) & 0
    \end{bmatrix}.
\end{align*}
We note that these matrices satisfy the following identities,
\begin{subequations}
\begin{align}
    &\operatorname{Tr}U_0=0,\label{eq:id1}\\
    &\operatorname{Tr}U_0U_1=0,\label{eq:id2}\\
    &U_0^2=0,\label{eq:id3}\\
    &U_0 U_1=-\frac{c_1}{c_2}U_1U_0,\label{eq:id4}\\
    &U_1^2=c_1c_2(c_2-c_1)U_1+c_1^3c_2^3 I.\label{eq:id6}
\end{align}
\end{subequations}
Making use of only these identities, it is a matter of expanding $\chi_p(z)$ around $\lambda=\infty$ and some combinatorics to obtain that its degree in $\lambda$ is less or equal to $3p$.
Before handling the general case, consider $p=3$. Expanding $\chi_3(z)$ around $\lambda=\infty$ we get
\begin{align*}
\frac{\chi_3(z)}{\lambda^{3\cdot 4}}=\operatorname{Tr}&\left[A(z+2)A(z+1)A(z)\right]\frac{1}{\lambda^{3\cdot 4}}\\\
    =\operatorname{Tr}&[
    U_0^3+\lambda^{-1}(U_0^2 U_1+U_0U_1U_0+U_1U_0^2)\\
    &+\lambda^{-2}(U_0^2U_2(z)+U_0U_2(z+1)U_0+U_2(z+2)U_0^2)\\
    &+\lambda^{-2}(U_0U_1^2+U_1U_0U_1+U_1^2U_0)+\mathcal{O}(\lambda^{-3})
    ]\\
    =\operatorname{Tr}&[0+\lambda^{-1}\cdot 0+\lambda^{-2}\cdot 0+3\lambda^{-2}U_0U_1^2    +\mathcal{O}(\lambda^{-3})    ]\\
    =\operatorname{Tr}&[3\lambda^{-2}(c_1c_2(c_2-c_1)U_0U_1+c_1^3c_2^3 U_0)+\mathcal{O}(\lambda^{-3})]\\
    =\hspace{-0.25mm}3\lambda&^{-2}(c_1c_2(c_2-c_1)\operatorname{Tr}[U_0U_1]+c_1^3c_2^3 \operatorname{Tr}[U_0])+\mathcal{O}(\lambda^{-3})]\\
    =\hspace{0.6mm}\mathcal{O}&(\lambda^{-3}),
\end{align*}
as $\lambda\rightarrow \infty$, where in the third equality we used identities \eqref{eq:id3} and $\operatorname{Tr}[XY]=\operatorname{Tr}[YX]$, in the fourth equality we used identity \eqref{eq:id6}, in the fifth equality we used linearity of $\operatorname{Tr}[\cdot]$ and in the final equality we used
identities \eqref{eq:id1} and \eqref{eq:id2}.
We conclude that $\chi_3(z)$, and hence in particular $\mathcal{I}_3$, is of total degree  $9$ in $(f,g)$, as desired.

For general $p\geq 3$, we note that expanding $\lambda^{-4 p}\chi_p(z)$ around $\lambda=\infty$, all the coefficients of $\lambda^{-m}$, $0\leq m\leq p-1$, are sums of traces
\begin{equation}\label{eq:traces}
    \operatorname{Tr}[S_{j_{p-1}}^{(p-1)}\cdot S_{j_{p-2}}^{(p-2)}\cdot\ldots\cdot S_{j_1}^{(1)}\cdot S_{j_0}^{(0)}],
\end{equation}
with sum of indices
\begin{equation}\label{eq:sumofindices}
    j_{p-1}+j_{p-2}+\ldots+j_{1}+j_0=m,
\end{equation}
where, for $0\leq k\leq p-1$,
\begin{equation*}
S_j^{(k)}=\begin{cases}
    U_0 & \text{if $j=0$},\\
    U_1 & \text{if $j=1$},\\
    U_j(z+k) & \text{if $j\geq 2$},
\end{cases}
\end{equation*}
that satisfy
\begin{equation}\label{eq:tracebound}
 \#\{k:j_k=0\}> \#\{k:j_k\geq 2\}.
\end{equation}
We claim that all such traces are equal to zero. Let us write $n=\#\{k:S_k=U_0\}$. Note that $n$ must be nonzero, otherwise the sum of indices in equation \eqref{eq:sumofindices} will be at least $p>m$.

Suppose that $n=1$, then, apart from one factor $U_0$, only factors $U_1$ appear in the product \eqref{eq:traces}. Using $\operatorname{Tr}[XY]=\operatorname{Tr}[YX]$, we can rewrite such a trace as $\operatorname{Tr}[U_0 U_1^{p-1}]$, which one can show is zero by recursive use of identity \eqref{eq:id6} and then applying \eqref{eq:id1} and \eqref{eq:id2}.

Next, we consider $n\geq 2$. Since $U_0^2=0$ and $U_0$ and $U_1$ almost commute, see \eqref{eq:id4}, it follows that a trace of the form \eqref{eq:traces} can only be nonzero if there are enough factors $U_j(z+k)$, $j\geq 2$, that can separate all the occurrences of $U_0$'s in the product. But, precisely because of inequality \eqref{eq:tracebound}, this is not the case. To be precise, by first applying $\operatorname{Tr}[XY]=\operatorname{Tr}[YX]$ if necessary, there will always be a group of consecutive factors in the product in equation \eqref{eq:traces},
\begin{equation*}
    S_{j_l}^{(l)}\cdot S_{j_{l-1}}^{(l-1)}\cdot\ldots\cdot S_{j_{k+1}}^{(k+1)}\cdot S_{j_{k}}^{(k)},
\end{equation*}
with $k<l$, where the indices in the formula should be read modulo $p$, that only consists of factors $U_0$ and $U_1$ with at least two occurrences of $U_0$. By identities \eqref{eq:id3} and \eqref{eq:id4}, such a product is equal to zero.
This proves the claim that all traces of the form \eqref{eq:traces} are zero. Therefore,
\begin{equation*}
    \lambda^{-4 p}\chi_p(z)=\mathcal{O}(\lambda^{-p}),
\end{equation*}
as $\lambda\rightarrow \infty$. Thus $\chi_p(z)$, and hence in particular $\mathcal{I}_p$, is of total degree  $3p$ in $(f,g)$, as desired. The theorem follows.
\end{proof}

\begin{remark}
Note that, in characteristic $p$, we have $D_t Q^p=0$ for any polynomial $Q\in \mathbb{F}_p[f,g,\alpha_1,\alpha_2,t]$. This means that $\mathcal{I}_p$ remains an integral of motion of $\Pfour$ upon dropping any terms that are $p$-th powers. We can furthermore always drop terms that only involve $\alpha_1$ and $\alpha_2$.
Applying this to $\mathcal{I}_2$ and $\mathcal{I}_3$, given in Appendix \ref{appendix:explicitformulas}, gives the respective simplified integrals
    \begin{align*}
    \change{\hat{\mathcal{I}}_2}=&+
    t f^2 g+t f g^2 +t^2 f g 
    +f g+\alpha_2 t f +\alpha_1 t g,\\
    \change{\hat{\mathcal{I}}_3}=&
    +t^3 f g+t^2 f^2 g-t^2 f g^2+\alpha_2 t^2 f+\alpha_1 t^2 g-\alpha_1 t f g+\alpha_2 t f g-\alpha_1 f^2 g\\
    &+\alpha_2 f^2 g-f^2 g+\alpha_1 f g^2-\alpha_2 f g^2- f g^2
    +\alpha_1\alpha_2  t-\alpha_1\alpha_2  f+\alpha_2^2 f-\alpha_2 f\\
    &-\alpha_1^2 g+\alpha_1 g+\alpha_1\alpha_2  g.
\end{align*}
We note, however, that these are no longer integrals of motion of \change{$\textrm{d}\Ptwo$}.
\end{remark}

\subsection{Symmetry group action on integrals of motion}\label{sec:sym}

To work out the action of the symmetry group on integrals of motion, we require explicit formulas for the coefficients $i_1$ and $i_p$ in Lemma \ref{lem:chi_form}. To obtain such formulas, we make use of the following
 lemma, proven in Appendix \ref{sec:appendix_proof}.
\begin{lemma}\label{lem:traceconstant}
    For any prime $p$, for $2\times 2$ matrix polynomials of degree $2$ in $z$,
    \begin{equation*}
        B(z)=B_0+z\,B_1+z^2\,B_2,
    \end{equation*}
     we have the following trace identity in characteristic $p$,
    \begin{gather*}
        \operatorname{Tr}\left[B(z+p-1)\cdot B(z+p-2)\cdot\ldots\cdot B(z+1)\cdot B(z)\right]=\\
        (z^p-z)^2 \operatorname{Tr}B_2^p-(z^p-z)\operatorname{Tr}B_2^{p-1}B_1+(z^p-z) \operatorname{Tr}B_1^p+\\       
        \operatorname{Tr}\left[B(p-1)\cdot B(p-2)\cdot\ldots\cdot B(1)\cdot B(0)\right].
    \end{gather*}
\end{lemma}

Using the above lemma, we obtain the following refinement of Lemma \ref{lem:chi_form}.
\begin{proposition}\label{prop:chi_form_ref}
 The polynomial $ \chi_p(z)$ takes the form
\begin{equation*}
    \chi_p(z)=\begin{cases}
        \mathcal{I}_p-t^p z+t^pz^p & \text{if $p\neq 2$,}\\
        \change{\mathcal{I}_2}+(1+t^2)(z+z^2) & \text{if $p= 2$,}
    \end{cases}
\end{equation*}
where $\mathcal{I}_p$ is defined in Definition \ref{def:integralofmotion}.   
\end{proposition}
\begin{proof}
The case $p=2$ is a matter of direct computation, so we assume that $p$ is an odd prime.
Using the notation in Lemma \ref{lem:chi_form}, we need to show that
$i_p=t^p$. To establish this, we claim that it is enough to prove that
\begin{enumerate}
    \item $i_p=\mathcal{O}(t)$ as $t\rightarrow 0$, and
    \item $i_p=t^p+\mathcal{O}(t^{p-1})$ as $t\rightarrow \infty$.
\end{enumerate}
Indeed, it follows from these two properties that
\begin{equation*}
    i_p=b_1 t+b_2 t^2+\ldots +b_{p-1}t^{p-1}+ t^p,
\end{equation*}
for some $b_k\in \mathbb{F}_p[f,g,\alpha_1,\alpha_2]$, $1\leq k\leq p-1$. Then, by equation \eqref{eq:Dt}, the derivative of $i_p$ must be zero, so
\begin{equation*}
0=b_1+t(D_tb_1+2b_2)+t^2(D_t b_2+3 b_3)+\ldots+t^{p-2}(D_t b_{p-2}+(p-1)b_{p-1})+t^{p-1}D_t b_{p-1}.
\end{equation*}
The constant term with respect to $t$ on the right-hand side is $b_1$, which must thus be zero. Then considering the coefficient of $t$, we find $b_2=0$ and continuing by induction we find $b_k=0$ for all $1\leq k\leq p-1$, so that indeed $i_p=t^p$.

Now (1) follows by applying Lemma \ref{lem:traceconstant} to the matrix
\begin{equation*}
    B(z)=A(z)|_{t=0},
\end{equation*}
with $w=1$, so that
\begin{equation*}
    B_0=A_0|_{t=0},\quad B_1=\begin{bmatrix}
        0 & 1\\
        h & 0
    \end{bmatrix},\quad B_2=\begin{bmatrix}
        0 & 0\\
        1 & 0
    \end{bmatrix}.
\end{equation*}
Indeed, since $B_2^2=0$ and $B_1^2=h I$,
\begin{equation*}
 \operatorname{Tr}B_2^p=0,\quad   \operatorname{Tr}B_2^{p-1}B_1=0,\quad \operatorname{Tr}B_1^p=0,
\end{equation*}
so that the formula in Lemma \ref{lem:traceconstant} shows that $\chi_p(z)|_{t=0}$ is just a constant in $z$. In other words, $i_p$ must vanish at $t=0$, that is, $i_p=\mathcal{O}(t)$ as $t\rightarrow 0$.

To prove (2), we note that
\begin{equation*}
    A(z)=(z-fg)\begin{bmatrix}
        (z+fg)t+\mathcal{O}(1) & w\\
        -w^{-1}(z+fg)^2t^2+\mathcal{O}(t) & -(z+fg)t+\mathcal{O}(1)
    \end{bmatrix}\qquad (t\rightarrow \infty),
\end{equation*}
and with induction on $n$, we obtain
\begin{gather*}
    M_n(z)=(z-fg)^{(n)} \\
    \cdot\begin{bmatrix}
        (z+fg)t^n+\mathcal{O}(t^{n-1}) & wt^{n-1}+\mathcal{O}(t^{n-2})\\
        -w^{-1}(z+fg)(z+fg+n-1)t^{n+1}+\mathcal{O}(t^n) & -(z+fg+n-1)t^n+\mathcal{O}(t^{n-1}) 
    \end{bmatrix},
\end{gather*}
as $t\rightarrow \infty$, where $(x)^{(n)}=x(x+1)\cdot\ldots\cdot(x+n-1)$ denotes the rising factorial, for $n\geq 1$.
Therefore
    \begin{equation*}
    \chi_p(z)=(z-fg)^{(p)}t^p+\mathcal{O}(t^{p-1})\qquad (t\rightarrow \infty).
\end{equation*}
Furthermore, we have the identity $(x)^{(p)}=x^p-x$ in characteristic $p$. In particular, $(x)^{(p)}$ is an $\mathbb{F}_p$-linear map and thus
    \begin{equation*}
    \chi_p(z)=t^p(fg-f^pg^p-z+z^p)+\mathcal{O}(t^{p-1})\qquad (t\rightarrow \infty).
\end{equation*}

Consequently,
\begin{equation*}
    i_0=t^p(fg-f^pg^p)+\mathcal{O}(t^{p-1}),\quad i_1=-t^p+\mathcal{O}(t^{p-1}),\quad i_p=t^p+\mathcal{O}(t^{p-1}),
\end{equation*}
as $t\rightarrow \infty$. 
This establishes (2) and as noted before, yields the proposition together with (1).
\end{proof}

We are now in position to work out the action of the affine Weyl symmetry group on the integrals of motion.
\begin{theorem}\label{thm:sym}
The generators $\{s_0,s_1,s_2,\pi\}$ of $\widetilde{W}(A_2^{(1)})$ act on the integral of motion $\mathcal{I}_{p}$ as
\begin{align*}
    s_0(\mathcal{I}_{p})&=\mathcal{I}_{p},\\
    s_1(\mathcal{I}_{p})&=\mathcal{I}_{p}+t^p(\alpha_1^p-\alpha_1),\\
    s_2(\mathcal{I}_{p})&=\mathcal{I}_{p}-t^p(\alpha_2^p-\alpha_2),\\
    \pi(\mathcal{I}_{p})&=\mathcal{I}_{p}-t^p(\alpha_2^p-\alpha_2),
\end{align*}
for odd primes $p$. For $p=2$, the formulas have to be adjusted by replacing $t^2$ by $1+t^2$ on the right-hand sides.
\end{theorem}
\begin{proof}
To prove the theorem, we require the actions of the symmetries on the coefficient matrix $A(z)$ given in equation \eqref{eq:coefficient_matrix}. We already know the actions of the translations $T_1$ and $T_2$ on $A(z)$, see equation \eqref{eq:translation_action_coef}. We claim that we have the following further actions of the symmetries on the coefficient matrix,
\begin{subequations}\label{eq:laxsymmetries}
\begin{align}
    T_1T_2s_0(A(z))&=R_{s_0}A(z)R_{s_0}^{-1},\label{eq:laxsymmetriess0}\\
    s_1(A(z))&=R_{s_1}A(z+\alpha_1)R_{s_1}^{-1},\label{eq:laxsymmetriess1}\\
    s_2(A(z))&=R_{s_2}A(z-\alpha_2)R_{s_2}^{-1},\label{eq:laxsymmetriess2}\\
    T_1\pi(A(z))&=R_{\hspace{0.5mm}\pi\hspace{0.5mm}}A(z-\alpha_2)R_{\hspace{0.5mm}\pi\hspace{0.5mm}}^{-1},\label{eq:laxsymmetriespi}
\end{align}
\end{subequations}
where
\begin{equation}\label{eq:Rmatrices}
    R_{s_0}=I,\quad R_{s_1}=\begin{bmatrix}
        1 & 0\\
        2t\alpha_1 w^{-1} & 1
    \end{bmatrix},\quad
    R_{s_2}=R_{\hspace{0.5mm}\pi\hspace{0.5mm}}=\begin{bmatrix}
        1 & 0\\
        -2t\alpha_2 w^{-1} & 1
    \end{bmatrix}.
\end{equation}
These equations can be verified by direct computation, but let us briefly explain where they come from.
Recall that the coefficient matrix is analytically characterised by its normalisation as $z\rightarrow \infty$ in equation \eqref{eq:coefmatrixnormalisation} as well as the value of its determinant in \eqref{eq:coefmatrixdet}. Let us denote the determinant of $A(z)$ by
\begin{equation*}
    \Delta(z)=-z(z-\alpha_1)(z+\alpha_2).
\end{equation*}
When we apply for example $s_1$ to this determinant, we obtain
\begin{equation*}
    s_1(\Delta(z))=-z(z+\alpha_1)(z+\alpha_1+\alpha_2)=\Delta(z+\alpha_1).
\end{equation*}
This suggests that $s_1(A(z))$ and $A(z+\alpha_1)$ are related. However, $A(z+\alpha_1)$ is not normalised as in equation \eqref{eq:coefmatrixnormalisation}. It can be brought into the correct normal form through conjugation by an appropriate lower-triangular matrix $R$. A direct calculation gives that $R=R_{s_1}$ as defined in the \eqref{eq:Rmatrices} does the job and relates $s_1(A(z))$ and $A(z+\alpha_1)$ through conjugation as in equation \eqref{eq:laxsymmetriess1}. The exact same procedure yields the formulas in \eqref{eq:laxsymmetries} for the remaining symmetries $T_1T_2s_0$, $s_2$ and $T_1\pi$.

From equations \eqref{eq:laxsymmetries}, we immediately obtain the actions of the symmetries on $M_p(z)$, defined in equation \eqref{eq:Mpdefi},
\begin{align*}
T_1T_2s_0(M_p(z))&=R_{s_0}M_p(z)R_{s_0}^{-1},\\
s_1(M_p(z))&=R_{s_1}M_p(z+\alpha_1)R_{s_1}^{-1},\\
s_2(M_p(z))&=R_{s_2}M_p(z-\alpha_2)R_{s_2}^{-1},\\
T_1\pi(M_p(z))&=R_{\hspace{0.5mm}\pi\hspace{0.5mm}}M_p(z-\alpha_2)R_{\hspace{0.5mm}\pi\hspace{0.5mm}}^{-1}.
\end{align*}
Then, by taking traces of both sides, and recalling from Lemma \ref{lem:chi_form} that $\chi_p(z)=\operatorname{Tr}(M_p(z))$ is invariant under $T_1$ and $T_2$, we get
\begin{align*}
    s_0(\chi_p(z))&=\chi_p(z),\\
    s_1(\chi_p(z))&=\chi_p(z+\alpha_1),\\
    s_2(\chi_p(z))&=\chi_p(z-\alpha_2),\\
    \pi(\chi_p(z))&=\chi_p(z-\alpha_2).
\end{align*}
Finally, recall from Proposition \ref{prop:chi_form_ref} that $\chi_p(z)=\mathcal{I}_p-t^pz+t^pz^p$ for $p\geq 3$, which yields the actions of the symmetries on $\mathcal{I}_{p}$ as described in the theorem for $p\geq 3$. When $p=2$, we use the alternative formula $\chi_2(z)=\mathcal{I}_2+(1+t^2)(z+z^2)$ to obtain the result, finishing the proof of the theorem.
\end{proof}

Theorem \ref{thm:invariance} is a corollary of the above theorem.
\begin{proof}[Proof of Theorem \ref{thm:invariance}.]
Define, for $p\neq 3$,
\begin{equation}
    \mathcal{L}_{p}=\begin{cases}
        -\tfrac{1}{3}t^p(\alpha_1^p-\alpha_1-\alpha_2^p+\alpha_2), & \text{if $p\geq 5$,}\\
        (1+t^2)(\alpha_1^2+\alpha_1+\alpha_2^2+\alpha_2) & \text{if $p=2$.}
    \end{cases}
\end{equation}
Then $\mathcal{L}_{p}$ satisfies the exact same equations as $\mathcal{I}_{p}$ in Theorem \ref{thm:sym}. Therefore, the difference
$\widetilde{\mathcal{I}}_{p}=\mathcal{I}_{p}-\mathcal{L}_{p}$ is invariant under the extended affine Weyl group symmetry, yielding the theorem.
\end{proof}

\section{Fibres and special solutions}\label{sec:fibres}

In this section, we investigate the fibres of the integrals of motion. To define these fibres properly, we first study the integrals of motions as functions on the initial value space of $\Pfour$ and \change{$\textrm{d}\Ptwo$}, introduced in the next section.

\subsection{Regularity of integrals of motion}\label{sec:regularity}

The fourth Painlev\'e equation and \change{$\textrm{d}\Ptwo$} can be simultaneously regularised on a rational surface of type $E_6^{(1)}$, called the initial value space \cites{okamoto1979,s:01}. The initial value space is well-defined over any field, regardless of its characteristic. 
We first recall its construction and then study the integral of motion as a function on the initial value space.

Fix an arbitrary field $\mathbf{k}$ and let $t\in \mathbf{k}$ and $(\alpha_1,\alpha_2)\in \mathbf{k}^2$.
Our starting point is the surface
\begin{equation}\label{eq:hirzebruch}
    \{(f,g)\in \mathbb{P}_\mathbf{k}^1\times \mathbb{P}_\mathbf{k}^1\}.
\end{equation}
This surface is covered by the four affine charts $(f,g)$, $(F,g)$, $(f,G)$ and $(F,G)$, where $F=1/f$ and $G=1/g$.
We perform eight consecutive blowups to this surface to obtain the initial value space. To this end, we use the following notation. In a given coordinate chart $(x,y)$, after blowing up the base point
\begin{equation*}
    b_i:\quad (x,y)=(x_*,y_*),
\end{equation*}
the corresponding exceptional line $E_i$ is covered by the coordinate charts $(u_i,v_i)$ and $(U_i,V_i)$, defined through
\begin{equation*}
    \begin{cases}
        x=x_*+u_i v_i, &\\
        y=y_*+v_i, &
    \end{cases}\qquad
    \begin{cases}
        x=x_*+U_i, &\\
        y=y_*+U_iV_i, &
    \end{cases}
\end{equation*}
so that the exceptional line has local equations $v_i=0$ and $U_i=0$ in these coordinate charts respectively. Denote by $\widehat{\mathcal{X}}_{\alpha}$ the result of blowing up the surface \eqref{eq:hirzebruch} at the following eight base points, 

\begin{equation*}\label{eq:basepoints}
\begin{aligned}
&b_1: (F,g)=(0,0), & &b_2: (U_1,V_1)=(0,-\alpha_2),\\
    &b_3: (f,G) = (0, 0),   &  &b_{4} : (u_3, v_3) = ( \alpha_1 ,0),\\
    &b_5 : (F,G) = (0,0), &   &b_{6} : (U_5, V_5) = (0,1),\\
    &b_{7} : (U_{6},V_{6})  = (0, -t),  & &b_{8} : (U_{7},V_{7}) = \left( 0,t^2+\alpha_1+\alpha_2-1\right).
\end{aligned}
\end{equation*}

Writing $b_i\leftarrow b_j$ when $b_j\in E_i$, we have three sequences of base points,
\begin{equation*}
   b_1\leftarrow b_2,\quad b_3\leftarrow b_4,\quad
   b_5\leftarrow b_6\leftarrow b_7\leftarrow b_8,
\end{equation*}
with the final base points being blown up to obtain $\widehat{\mathcal{X}}_{\alpha}$ given by $b_2$, $b_4$ and $b_8$ respectively. Let $D\subseteq \widehat{\mathcal{X}}_{\alpha}$ denote the strict transform under the final blow-ups of the total transform of the curve $\{f=\infty\}\cup\{g=\infty\}$ under the non-final blow ups. \change{This is the support of} the unique, effective, anti-canonical divisor, represented in blue in Figure \ref{fig:initialvaluespace}. We refer to \cite{s:01} for more details.


\begin{definition}\label{defi:initialvaluespace}
We call $\mathcal{X}_{\alpha}:=\widehat{\mathcal{X}}_{\alpha}\setminus D$ the initial value space of $\Pfour$ and \change{$\textrm{d}\Ptwo$} over $\mathbf{k}$. Here the dependence of $\mathcal{X}_{\alpha}$ on $t$ is suppressed and we will sometimes also suppress the dependence on $\alpha$, simply writing $\mathcal{X}$ for the initial value space.
\end{definition}

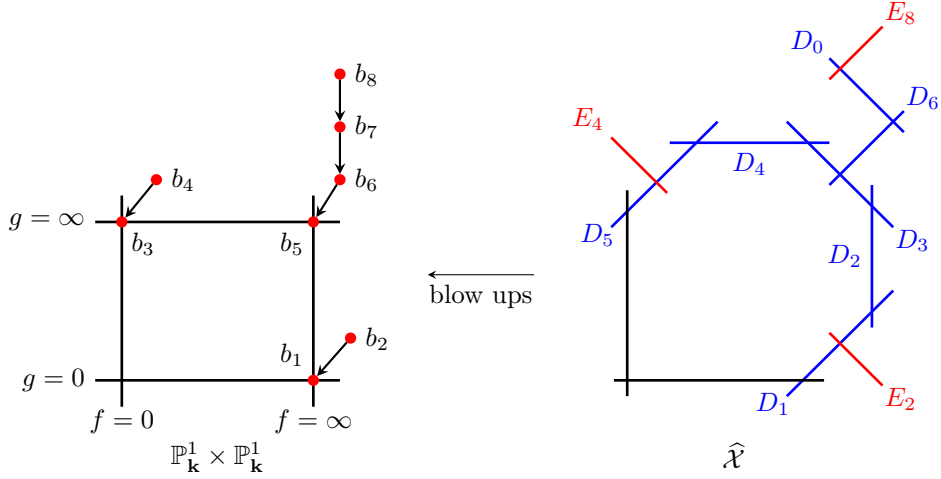
\begin{figure}[t]
\centering
	\begin{tikzpicture}[scale=.70,>=stealth,basept/.style={circle, draw=red!100, fill=red!100, thick, inner sep=0pt,minimum size=1.2mm}]
		\begin{scope}[xshift = -4cm, yshift = -1cm]
			\draw [black, line width = 1pt] 	(4.1,2.5) 	-- (-0.5,2.5)	node [left]  {$g=\infty$} node[pos=0, right] {};
			\draw [black, line width = 1pt] 	(0,3) -- (0,-1)			node [below = -1mm] {$f=0$}  node[pos=0, above, xshift=-7pt] {} ;
			\draw [black, line width = 1pt] 	(3.6,3) -- (3.6,-1)		node [below = -1mm]  {$f=\infty$} node[pos=0, above, xshift=7pt] {};
			\draw [black, line width = 1pt] 	(4.1,-.5) 	-- (-0.5,-0.5)	node [left]  {$g=0$} node[pos=0, right] {};

			\node (p1) at (3.6,-.5) [basept,label={[xshift=-8pt, yshift = -0.2 pt] $b_{1}$}] {};
			\node (p2) at (4.3,0.3) [basept,label={[xshift=10pt, yshift = -10 pt] $b_{2}$}] {};
			\node (p3) at (0,2.5) [basept,label={[yshift=-18pt, xshift=+8pt] $b_{3}$}] {};
			\node (p4) at (0.65,3.3) [basept,label={[xshift=10pt, yshift = -10 pt] $b_{4}$}] {};
			\node (p5) at (3.6,2.5) [basept,label={[xshift=-8pt,yshift=-18pt] $b_{5}$}] {};
			\node (p6) at (4.1,3.3) [basept,label={[xshift=10pt, yshift = -10 pt] $b_{6}$}] {};
			\node (p7) at (4.1,4.3) [basept,label={[xshift=10pt, yshift = -10 pt] $b_{7}$}] {};
			\node (p8) at (4.1,5.3) [basept,label={[xshift=10pt, yshift = -10 pt] $b_{8}$}] {};
			\draw [line width = 0.8pt, ->] (p2) -- (p1);
			\draw [line width = 0.8pt, ->] (p4) -- (p3);
			\draw [line width = 0.8pt, ->] (p6) -- (p5);
			\draw [line width = 0.8pt, ->] (p7) -- (p6);
			\draw [line width = 0.8pt, ->] (p8) -- (p7);
			\node (P1P1) at (1.8, -1.5) [below] {$\mathbb{P}_\mathbf{k}^1 \times \mathbb{P}_\mathbf{k}^1$};
		\end{scope}
	
		\draw [->] (3.75,0.5)--(1.75,0.5) node[pos=0.5, below] {$\text{blow ups}$};
	
		\begin{scope}[xshift = 6.5cm, yshift= .5cm]
			\draw [blue, line width = 1pt] 	(2.8,2.5) 	-- (-0.2,2.5)	node [pos = .5, below]  {
			$D_4$
			} node[pos=0, right] {};
			\draw [blue, line width = 1pt] 	(3.6,1.7) -- (3.6,-1)		node [pos = .5, left]  {
			$D_2$
			} node[pos=0, above, xshift=7pt] {};

			\draw [blue, line width = 1pt] 	(-1.3,0.9) 	-- (.7, 2.9)	 node [left] {} node[pos = 0, below left = -2.5mm] {
			$D_5$
			};
			\draw [red, line width = 1pt] 	(-.25,1.55) 	-- (-1.3,2.6)	 node [left] {} node[pos=1, above left = -0.5mm ] {$E_4$};
			
			\draw [blue, line width = 1pt] 	(4,-.3) node[left]{}	-- (2,-2.3)	 node [below left] {} node[below left = -2.5mm] {
			$D_1$
			};

			\draw [red, line width = 1pt] 	(2.8,-1.1) 	-- (3.8,-2.1)	 node [left] {} node[pos=1, below right  = -1.5mm ] {$E_2$};
			
			\draw [blue, line width = 1pt]	(4,.9) -- (2,2.9) node[ pos=0, below right  = -1.5mm] {
			$D_3$
			};
			\draw [blue, line width = 1pt]	(2.8,1.7) -- (4.2,3.1) node [above right = -1.5mm] {
			$D_6$
			} ;
			\draw [blue, line width = 1pt]	(4.2,2.7) -- (2.8,4.1) node [above left = -1mm] {
			$D_0$
			};
			\draw [red, line width = 1pt]	(2.8,3.7) -- (3.8,4.7)  node [above right = -1.5mm] {$E_8$};

			\draw [black, line width = 1pt]  (-1,1.6) -- (-1,-2.3);
			\draw [black, line width = 1pt]  (-1.25,-2) 	-- (2.7,-2);
\node (Xa) at (1, -3) [below] {$\widehat{\mathcal{X}}$};

		\end{scope}
	\end{tikzpicture}
	\caption{Representation of blow ups leading to the initial value space $\mathcal{X}=\widehat{\mathcal{X}}\setminus D$, with in blue \change{the support $D$ of the effective anti-}\change{canonical divisor}.}
	\label{fig:initialvaluespace}
\end{figure}

The bi-rational mappings in \eqref{eq:painleve_equation} lift uniquely to isomorphisms
\begin{equation}\label{eq:translation_isos}
    T_j:\widehat{\mathcal{X}}_{\alpha}\rightarrow \widehat{\mathcal{X}}_{\overline{\alpha}^j}\qquad (j=1,2),
\end{equation}
 except in the degenerate parameter cases $\alpha_1\in \{0,-1\}$, $\alpha_2\in \{0,-1\}$ and $\alpha_1+\alpha_2\in \{0,-1\}$. They furthermore map \change{the support of} the anti-canonical divisor on the left to \change{the support of} the anti-canonical divisor on the right, inducing corresponding maps between the initial value spaces $\mathcal{X}_{\alpha}$ and $\mathcal{X}_{\overline{\alpha}^j}$.

Starting with a regular value $(f,g)\in \mathbb{A}_\mathbf{k}^2$, it is possible to land on the exceptional lines $E_2$, $E_4$ and $E_8$ after a number of translations. In what follows, we require coordinates that parametrise these exceptional lines away from the \change{support of the} anticanonical divisor. 
\begin{itemize}
    \item The part of $E_2$ away from $D$ is covered by the coordinates $(x_2,y_2)=(U_2,V_2)$, given by
\begin{equation*}
    f=\frac{1}{x_2},\qquad g=x_2(-\alpha_2+x_2y_2),
\end{equation*}
in which $E_2$ has local equation $x_2=0$.
 \item The part of $E_4$ away from $D$ is covered by the coordinates $(x_4,y_4)=(v_4,u_4)$, given by
\begin{equation*}
    f=x_4(\alpha_1+x_4y_4),\qquad g=\frac{1}{x_4},
\end{equation*}
in which $E_4$ has local equation $x_4=0$.
 \item The part of $E_8$ away from $D$ is covered by the coordinates $(x_8,y_8)=(U_8,V_8)$, given by
\begin{equation*}
    f=\frac{1}{x_8},\qquad g=\frac{1}{x_8(1-t\,x_8+(t^2+1-\alpha_1-\alpha_2)x_8^2+y_8\,x_8^3)},
\end{equation*}
in which $E_8$ has local equation $x_8=0$.
\end{itemize}
In particular, we note that the part of $E_j$ away from $D$ is given by the affine line $\{y_j\in \mathbb{A}_\mathbf{k}^1, x_j=0\}$, for $j=2,4,8$.

\begin{theorem}\label{thm:morphism}
In characteristic $p\geq 2$,
the integral of motion $\mathcal{I}_p$ extends polynomially to the exceptional curves $E_2$, $E_4$ and $E_8$. Namely, for $j=2,4,8$, in the local coordinates $(x_j,y_j)$ introduced above,
\begin{equation*}
  \mathcal{I}_p=P_j(y_j)+x_j R_j(x_j,y_j),
\end{equation*}
where $P_j\in \mathbb{F}_p[y_j,\alpha_1,\alpha_2,t]$ is a polynomial and $R_j\in \mathbb{F}_p(x_j,y_j,\alpha_1,\alpha_2,t)$ is a rational function which in reduced form does not have $x_j$ as a divisor of the denominator. In particular, the integral of motion defines a regular map
\begin{equation}\label{eq:integral_morphism}
   \mathcal{I}_p: \mathcal{X}\rightarrow \mathbb{A}_\mathbf{k}^{1},
\end{equation}
for any choice of parameters.
\end{theorem}
\begin{proof} 
 Let us consider the integral of motion on $E_2$. By direct substitution of the coordinates $(x_2,y_2)$ that cover $E_2$ away from the \change{support of the} anti-canonical divisor, into the coefficient matrix of the spectral equation, we get
\begin{equation*}
 A(z)=A_{E_2}(z)+\mathcal{O}(x_2),   
\end{equation*}
as $x_2\rightarrow 0$, where the error term is polynomial in $z,x_2,y_2,\alpha_1,\alpha_2,t$, and
\begin{align*}
    A_{E_2}(z)&=\begin{bmatrix}
        t(z-\alpha_2)(z+\alpha_2) & w(z+\alpha_2)\\
        w^{-1}Q_{E_2}(z)  & -t(z-\alpha_2)(z+\alpha_2)-y_2
    \end{bmatrix}\\
    Q_{E_2}(z)&=-t^2 z^3+(1+\alpha_2t^2)z^2+(t^2\alpha_2^2-\alpha_1-t\,y_2)z+\alpha_2 t(y_2-\alpha_2^2 t).
\end{align*}
In other words, under the parametrisation of the coefficient matrix $A(z)$, points on $E_2\setminus D$ correspond to coefficient matrices $A(z)$ such that the first row vanishes at $z=-\alpha_2$. It follows that
\begin{align*}
    \mathcal{I}_p&=P_2(y_2)+\mathcal{O}(x_2),\\
    P_2(y_2)&=\operatorname{Tr}[A_{E_2}(p-1)\cdot A_{E_2}(p-2)\cdot\ldots\cdot A_{E_2}(1)\cdot A_{E_2}(0)],
\end{align*}
as $x_2\rightarrow 0$, where the error term $\mathcal{O}(x_2)$ is a polynomial in $x_2,y_2,\alpha_1,\alpha_2,t$. This proves the theorem for the case $j=2$, with $R_2(x_2,y_2)$ a polynomial.

Similarly, by direct substitution of the coordinates $(x_4,y_4)$ that cover $E_4\setminus D$, we get
\begin{equation*}
 A(z)=A_{E_4}(z)+\mathcal{O}(x_4),   
\end{equation*}
as $x_4\rightarrow 0$, where
\begin{align*}
    A_{E_4}(z)&=\begin{bmatrix}
        t(z-\alpha_1)(z+\alpha_1)+y_4 & w(z-\alpha_1)\\
        w^{-1}Q_{E_4}(z)  & -t(z-\alpha_1)(z+\alpha_1)
    \end{bmatrix}\\
    Q_{E_4}(z)&=-t^2 z^3+(1-\alpha_1t^2)z^2+(t^2\alpha_1^2+\alpha_2-t\,y_4)z-\alpha_1 t(y_4-\alpha_1^2 t).
\end{align*}
This yields the theorem for $j=4$ as above, with $R_4(x_4,y_4)$ a polynomial.

We cannot employ the same method for $E_8$, since there the coefficients of $A(z)$ blow up. Instead, we use the fact that $\mathcal{I}_{p,\alpha}$ is an integral of motion, which in particular gives
\begin{equation}\label{eq:integral_equality}
    \mathcal{I}_{p,(\alpha_1,\alpha_2)}(f,g)=\mathcal{I}_{p,(\alpha_1+1,\alpha_2)}(\overline{f}^1,\overline{g}^1).
\end{equation}
We are interested in the left-hand side for $(f,g)\in E_8\setminus D$,  but this is equivalent to $(\overline{f}^1,\overline{g}^1)\in \overline{E}^1_4\setminus \overline{D}^1$. Indeed, a direct calculation yields that the coordinate chart $(x_8,y_8)$ covering $E_8\setminus D$ and the coordinate chart $(\overline{x}^1_4,\overline{y}^1_4)$ covering $\overline{E}^1_4\setminus\overline{D}^1$ are related by
\begin{equation*}
\overline{x}^1_4=-x_8+\mathcal{O}(x_8^2),\qquad \overline{y}^1_4=-y_8-2t+3t(\alpha_1+\alpha_2)-t^3+\mathcal{O}(x_8),
\end{equation*}
as $x_8\rightarrow 0$, and, conversely,
\begin{equation*}
x_8=-\overline{x}^1_4 +\mathcal{O}((\overline{x}^1_4)^2),\qquad
    y_8=-\overline{y}^1_4-2t+3t(\alpha_1+\alpha_2)-t^3+\mathcal{O}(\overline{x}^1_4),
\end{equation*}
as $\overline{x}^1_4\rightarrow 0$. Since we already know that the theorem holds for the right-hand side, i,e, when $j=4$, it thus also holds for the left-hand side, i.e. $j=8$. However, since the change of coordinates between $(x_8,y_8)$ and $(\overline{x}^1_4,\overline{y}^1_4)$ is birational, this time, $R_8(x_8,y_8)$ is a rational function.
This finishes the proof of the theorem.
\end{proof}

\begin{remark}\label{rem:conjhigherorder}
    By looking at explicit formulas of $\mathcal{I}_p$ for different values of $p$, we put forward the following conjectural formula for its highest order terms,
\begin{equation*}
   \mathcal{I}_p=\begin{cases}
   (fg^2-f^2g-t fg
   -\alpha_2f-\alpha_1g)^p+O(f,g) & \text{if $p\neq 2$,}\\
   (fg^2+f^2g+t fg
   +\alpha_2f+\alpha_1g)^2+t fg(f+g)+O(f,g)  & \text{if $p=2$,}
   \end{cases}
\end{equation*}
where $O(f,g)$ is a polynomial of degree at most $p-1$ in $f$ and at most $p-1$ in $g$. We checked this assertion in Magma for all primes $2\leq p \leq 53$. This conjectural formula suggests that \eqref{eq:integral_morphism} extends to a regular map from $\widehat{\mathcal{X}}$ to $\mathbb{P}_\mathbf{k}^1$, with the effective anti-canonical divisor given by the \change{fibre $\mathcal{I}_p^{-1}(\infty)$}.
\end{remark}

\subsection{Reducible fibres and reductions to Riccati equations}\label{sec:fibre_reducible}

In this section, we study fibres of the map $\mathcal{I}_p:\mathcal{X}\rightarrow \mathbb{A}_\mathbf{k}^1$, see equation \eqref{eq:integral_morphism}. We relate reducible fibres to special solutions of $\Pfour$ and \change{$\textrm{d}\Ptwo$}.

The fourth Painlev\'e equation over $\mathbb{C}$  admits classical function solutions when the parameters lie on an edge of a Weyl chamber, that is, one or more of $\alpha_1,\alpha_2$ and $\alpha_1+\alpha_2$ is integer \cite{noumibook}.

For example,
setting $\alpha_1=0$ and $f=0$, $\Pfour$ reduces to the following Riccati equation for $g$,
\begin{equation}\label{eq:riccati}
    D_t g=g(t-g)+\alpha_2.
\end{equation}
This equation is linearised by writing $g$ as a logarithmic derivative,
\begin{equation}\label{eq:linearisation}
   g=\frac{D_t r}{r},\qquad  D_t^2 r= t D_t r+\alpha_2 r.
\end{equation}
The general solution of the latter linear ODE can be expressed in terms of parabolic cylinder functions.

Next, we work in a finite characteristic $p$ and consider the corresponding integral of motion $\mathcal{I}_p$ when $\alpha_1=0$. In that case, the determinant of the coefficient matrix $A(z)$, see equation \eqref{eq:coefmatrixdet}, is
\begin{equation*}
    |A(z)|=-z^2(z+\alpha_2).
\end{equation*}
For the special value $f=0$, $z$ is an overall factor of the coefficient matrix $A(z)$. That is, the curve $f=0$ in the initial value space $\mathcal{X}$ parametrises those coefficient matrices $A(z)$ with $A(0)=0$. In particular, on this whole curve, the integral of motion $\mathcal{I}_{p}=0$. This means that $f$ must be a divisor of $\mathcal{I}_p|_{\alpha_1=0}$ in the polynomial ring $\mathbb{F}_p(t,\alpha_2)[f,g]$.
In particular, when $\alpha_1=0$, the integral of motion is a reducible polynomial in $(f,g)$. Since the integral of motion is furthermore invariant under translations, the same holds true whenever $\alpha_1\in\mathbb{F}_p\subseteq \mathbf{k}$.

Similarly, when $\alpha_2=0$, the curve $g=0$ in the initial value space $\mathcal{X}$ parametrises those coefficient matrices $A(z)$ that have an overall factor $z$, that is, $A(0)=0$. So, in this case, $g$ is a divisor of $\mathcal{I}_p|_{\alpha_2=0}$ and again, for general $\alpha_2\in\mathbb{F}_p\subseteq \mathbf{k}$, $\mathcal{I}_{p}$ is a reducible polynomial in $(f,g)$ over $\mathbb{F}_p(t,\alpha_1)$.

Finally, consider $\alpha_2=-\alpha_1$, so that the determinant of the coefficient matrix becomes
\begin{equation*}
    |A(z)|=-z(z-\alpha_1)^2.
\end{equation*}
In this case, the curve $fg-\alpha_1=0$ in the initial value space $\mathcal{X}$ parametrises those coefficient matrices $A(z)$ that have an overall factor $z-\alpha_1$, i.e. $A(\alpha_1)=0$. In this case, again one of the fibres of $\mathcal{I}_{p}$ is reducible. Computations in Magma suggest that these are all the ways that reducible fibres may appear in $\mathcal{I}_{p}$. Namely, we have the following conjecture.

\begin{conjecture}\label{conj:irred}
Let $\mathbf{k}$ be a field of finite characteristic $p>0$, then the morphism $\mathcal{I}_p:\mathcal{X}\rightarrow \mathbb{A}_\mathbf{k}^1$, see equation \eqref{eq:integral_morphism},
has a reducible fibre if and only if $\alpha_1\in\mathbb{F}_p$, $\alpha_2\in\mathbb{F}_p$ or $\alpha_1+\alpha_2\in\mathbb{F}_p$. In such case, there is only one reducible fibre, given by
\begin{equation*}
    \mathcal{I}_p^{-1}(c),\quad c=\begin{cases}
    0 & \text{if $\alpha_1\in\mathbb{F}_p$ or $\alpha_2\in\mathbb{F}_p$,}\\
        (\alpha_1-\alpha_1^p)t^p & \text{if }\alpha_1+\alpha_2\in \mathbb{F}_p, p\neq 2,\\
        (\alpha_1+\alpha_1^2)(1+t)^2 & \text{if }\alpha_1+\alpha_2\in \mathbb{F}_p, p= 2.
    \end{cases}
\end{equation*}
\end{conjecture}

We checked the statement of the conjecture in Magma for $\mathbf{k}=\mathbb{F}_q$, with 
\begin{equation*}
    q\in\{2^5,3^4,5^2,7^2,11,13,17,19,23,29\},
\end{equation*}
 for all combinations of parameter values of $\alpha_1,\alpha_2,t\in \mathbf{k}$.

As an example, setting $\alpha_1=1$, the integral $\mathcal{I}_3$ factorises as
\begin{align*}
   \mathcal{I}_3=&(2 \alpha_2+2 t g+2 f g+g^2)\times (\alpha_2^2 f^3+2 \alpha_2 t f^3 g+2 \alpha_2 f^4 g+\alpha_2 f^3 g^2+2 \alpha_2 f\\
   &+t^2 f^3 g^2+2 t^2 f+2 t f^4 g^2+t f^3 g^3+2 t+f^5 g^2+f^4 g^3+f^3 g^4+2 f+2 g).
\end{align*}

\subsection{Rational solutions} \label{sec:ratsol}
\change{The fourth Painlev\'e equation admits families of rational solutions expressible in terms of generalised Hermite polynomials and a family of rational solutions expressible in terms of generalised Okamoto polynomials \cites{BCH95,clarkson03,noumiyamada99,okamoto86}. In this section, we study reductions modulo primes of one of the families of Hermite rationals and also comment on reductions of the family of Okamoto rationals.}

We again set $\alpha_1=0$ and $f=0$, so that $\Pfour$ reduces to the Riccati equation \eqref{eq:riccati}, which can be linearised as in \eqref{eq:linearisation}.
When $\alpha_2=-n\in\mathbb{Z}_{\leq 0}$, the linear ODE in \eqref{eq:linearisation} has the $n$th probabilist's Hermite polynomial as a solution,
\begin{equation*}
    r=H_n,\quad H_n=(-1)^n e^{\frac{t^2}{2}}D_t^n e^{-\frac{t^2}{2}}.
\end{equation*}
This yields the semi-infinite family
\begin{equation}\label{eq:famedge1}
    (\alpha_1,\alpha_2,f,g)=\left(0,-n,0,\frac{D_t H_n}{H_n}\right),\qquad (n\in\mathbb{Z}_{\geq 0}),
\end{equation}
of solutions to $\Pfour$. By repeated application of the translation $T_1^{-1}$ to its members, we obtain the family of Hermite rationals,
\begin{equation}\label{eq:famHermite}
    (\alpha_1,\alpha_2,f,g)=\left(-m,-n,f_{m,n},g_{m,n}\right),\qquad (m,n\in\mathbb{Z}_{\geq 0}),
\end{equation}
expressible in terms of Wronskians of Hermite polynomials \cites{noumiyamada99,okamoto86}.

Similar to the subfamily \eqref{eq:famedge1} for $m=0$, we have the following subfamily when $n=0$,
\begin{equation}\label{eq:famedge2}
    (\alpha_1,\alpha_2,f,g)=\left(-m,0,\frac{D_t \widetilde{H}_n}{\widetilde{H}_n},0\right),\qquad (m\in\mathbb{Z}_{\geq 0}),
\end{equation}
where
\begin{equation*}
    \widetilde{H}_n=i^{-n}H_n|_{t\mapsto i t}.
\end{equation*}
Here we note that both $H_n$ and $\widetilde{H}_n$ are monic of degree $n$ with \change{integer coefficients}. Fixing any prime $p$, we thus have well-defined reductions of these Hermite polynomials modulo $p$.

For the sake of simplicity, assume $p\neq 2$. Then we have the following identity in characteristic $p$,
\begin{equation*}
    (D_t+t)^p=D_t^p+t^p,
\end{equation*}
among operators, see e.g. \cite{bavula}.
Since $e^{-\frac{t^2}{2}}$ formally lies in the kernel of $D_t+t$, we get
\begin{equation*}
    D_t^p e^{-\frac{t^2}{2}}=-t^p e^{-\frac{t^2}{2}},
\end{equation*}
and therefore $H_{n+p}=t^p H_n$ and $\widetilde{H}_{n+p}=t^p \widetilde{H}_n$ for $n\geq 0$. In particular,
  logarithmic derivatives of these Hermite polynomials in characteristic $p$ are periodic in the index with period $p$,
  \begin{equation*}
   \frac{D_t H_{n+p}}{H_{n+p}}=\frac{D_t H_n}{H_n},\qquad    \frac{D_t \widetilde{H}_{n+p}}{\widetilde{H}_{n+p}}=\frac{D_t \widetilde{H}_n}{\widetilde{H}_n}\qquad (n\geq 0).
  \end{equation*}
This means that the two boundary subfamilies \eqref{eq:famedge1} and \eqref{eq:famedge2} have well-defined reductions modulo $p$. 

Since the translations $T_1$ and $T_2$, as well as their inverses, given by
\begin{equation*}
\begin{aligned}
    \underline{f}_1&=-\frac{\alpha_1+\alpha_2-1+g (f-g+t)}{f-g+t}, & 
     \underline{g}_1&=+\frac{(f-g+t) (\alpha_2+g (f-g+t))}{\alpha_1+\alpha_2-1+g (f-g+t)},\\
     \underline{g}_2&=-\frac{\alpha_1+\alpha_2-1+f (f-g+t)}{f-g+t}, &     \underline{f}_2&=-\frac{(f-g+t) (\alpha_1+f (f-g+t))}{\alpha_1+\alpha_2-1+f (f-g+t)},
\end{aligned}
\end{equation*}
where $\underline{(\cdot)}_j=T_j^{-1}(\cdot)$, for $j=1,2$, are also well-defined modulo $p$, it then follows that the whole family of Hermite rationals \eqref{eq:famHermite} has a good reduction modulo $p$. 

Denoting, for any rational function $h\in \mathbb{Q}(t)$, its reduction modulo $p$ by $\tilde{h}\in \mathbb{F}_p(t)$, when well-defined,
 we have the following periodicity properties for the family of Hermite rationals \eqref{eq:famHermite},
\begin{equation*}
    \tilde{f}_{m+p,n}=\tilde{f}_{m,n+p}=\tilde{f}_{m,n},\qquad \tilde{g}_{m+p,n}=\tilde{g}_{m,n+p}=\tilde{g}_{m,n}\qquad  (m,n\in\mathbb{Z}_{\geq 0}).
\end{equation*}

In other words, we obtain the reduced family of rational solutions
\begin{equation}\label{eq:famHermitered}
    (\alpha_1,\alpha_2,f,g)=\left(-m,-n,\tilde{f}_{m,n},\tilde{g}_{m,n}\right),\qquad (m,n\in\mathbb{F}_p),
\end{equation}
of $\Pfour$ over $\mathbf{k}=\mathbb{F}_p$.

We further note that a similar calculation gives the same result when $p=2$. In that case, the family effectively only has four members, given in Table \ref{tab:hermitep2}.
\begin{table}[H]
    \centering
    \begin{tabular}{c|c|c}
     $m$\textbackslash $n$    & $0$ & $1$  \\\hline
    $0$  & $(0,0)$ & $\left(0,\frac{1}{t}\right)$  \\\hline
    $1$ & $\left(\frac{1}{t},0\right)$ & $\left(\frac{1}{t},\frac{1}{t}\right)$
    \end{tabular}
    \caption{Formulas for members $(\tilde{f}_{m,n},\tilde{g}_{m,n})$ of \eqref{eq:famHermitered} with $p=2$.}
    \label{tab:hermitep2}
\end{table}
When both $\alpha_1=0$ and $\alpha_2=0$, and thus also $\alpha_1+\alpha_2=0$, the three corresponding reducible fibres, see Conjecture \ref{conj:irred}, coalesce and $\mathcal{I}_2$ factors over $\mathbb{F}_2$ into three irreducible factors,
\begin{equation*}
  \mathcal{I}_2=f\, g\, (f^3 g+f g^3+t^2 f g +
    t f +t\,  g +t^2 +1).
\end{equation*}
In particular, we see that two of the factors vanish for the seed solution $(f,g)=(0,0)$ of the family \eqref{eq:famHermitered} with $p=2$.

The same holds true for any prime $p$, since we know that $f$ and $g$ are factors of $\mathcal{I}_p$ when $\alpha_1=\alpha_2=0$, as explained in Section \ref{sec:fibre_reducible}, and the seed solution of \eqref{eq:famHermite} is $(f,g)=(0,0)$. More generally, for any choice of $(\alpha_1,\alpha_2)=(-m,-n)\in\mathbb{F}_p^2$, the integral admits a factorisation into three factors, and two of them vanish identically in $t$ at $(f,g)=(\tilde{f}_{m,n},\tilde{g}_{m,n})$. We discuss examples of this for $p=3$ and $p=5$.


When $p=3$, the family \eqref{eq:famHermitered} has nine members, given in Table \ref{tab:hermitep3}.
\begin{table}[H]
    \centering
    \begin{tabular}{c|c|c|c}
     $m$\textbackslash $n$    & $0$ & $1$ & $2$ \\\hline
    $0$  & $(0,0)$ & $\left(0,\frac{1}{t}\right)$ & $\left(0,\frac{2t}{t^2+2}\right)$  \\\hline
    $1$ & $\left(\frac{1}{t},0\right)$ & $\left(\frac{t^2+1}{t(t^2+2)},\frac{t^2+2}{t(t^2+1)}\right)$ & $\left(\frac{t}{t^2+2},\frac{2(t^2+1)}{t(t^2+2)}\right)$  \\\hline
    $2$     & $\left(\frac{2t}{t^2+1},0\right)$ & $\left(\frac{2(t^2+2)}{t(t^2+1)},\frac{t}{t^2+1}\right)$  & $\left(\frac{2}{t},\frac{2}{t}\right)$
    \end{tabular}
    \caption{Formulas for members $(\tilde{f}_{m,n},\tilde{g}_{m,n})$ of \eqref{eq:famHermitered} with $p=3$.}
   \label{tab:hermitep3}
\end{table}
Setting $m=1$ and $n=2$, so that $\alpha_1=2$ and $\alpha_2=1$, the integral $\mathcal{I}_3$ factorises as
\begin{align*}
    \mathcal{I}_3=&2 (f g+1) \left(f^2-f g+t f+2\right)\times\\
    &\left(f^3 g^2+f^2 g^3+f g^4-t f^2 g^2+t f g^3+t^2 f g^2- f^2 g+g^3+t f g+f- t^2 g-t\right).
\end{align*}
The last two factors vanish simultaneously only at $(f,g)=(\tilde{f}_{1,2},\tilde{g}_{1,2})$.

When $p=5$, the family \eqref{eq:famHermitered} has twenty-five members, with a selection of them given in Table \ref{tab:hermitep5}.

\begin{table}[H]
    \centering
    \begin{tabular}{|c|c|c|}\hline
     $m$\textbackslash $n$    & $0$ & $1$ \\\hline
    $0$  & $(0,0)$ & $\left(0,\frac{1}{t}\right)$ \\\hline
    $1$  & $\left(\frac{1}{t},0\right)$ & $\left(\frac{t^2+1}{t (t^2+4)},\frac{t^2+4}{t (t^2+1)}\right)$  \\\hline
    $2$  & $\left(\frac{2t}{t^2+1},0\right)$ & $\left(\frac{2t(t^2+3)(t^2+4)}{(t^2+1)(t^4+3)},\frac{t^4+3}{t(t^2+1)(t^2+3)}\right)$  \\\hline
    $3$  & $\left(\frac{3(t^2+1)}{t(t^2+3)},0\right)$ & $\left(\frac{3(t^2+2)(t^2+4)(t^4+3)}{t(t^2+3)(t^6+3 t^4+4 t^2+1)},\frac{t^6+3 t^4+4 t^2+1}{t(t^2+2)(t^2+3)(t^2+4)}\right)$   \\\hline
    $4$  & $\left(\frac{4t(t^2+3)}{(t^2+2)(t^2+4)},0\right)$ & $\left(\frac{4(t^6+3 t^4+4 t^2+1)}{t(t^2+2)(t^2+3)(t^2+4)},\frac{t(t^2+3)}{(t^2+2)(t^2+4)}\right)$ \\\hline
    \end{tabular}
    \caption{Formulas for a selection of the members $(\tilde{f}_{m,n},\tilde{g}_{m,n})$ of \eqref{eq:famHermitered} with $p=5$\change{, where we wrote all the rational functions in} \change{factored form over $\mathbb{F}_5$}.}
   \label{tab:hermitep5}
\end{table}


Setting $m=2$ and $n=1$, so that $\alpha_1=3$ and $\alpha_2=4$, the integral $\mathcal{I}_5$ factorises as
\begin{align*}
    \mathcal{I}_5=\,&4 (f^2 g^2+4 f g^3+f g^2 t+3 f g+3 g^2+4 g t+1)\times \\
    &(f^3 g+3 f^2 g^2+2 f^2 g t+4 f^2+f g^3+3 f g^2 t+f g t^2+4 f g+2 f t+2 g^2+3 t^2+1)\times \\
    &(f^5 g^2+3 f^4 g^3+2 f^4 g^2 t+3 f^4 g+f^3 g^4+3 f^3 g^3 t+f^3 g^2 t^2+f^3 g^2+f^3 g t+f^3\\
    &\,+f^2 g^3+f^2 g^2 t+3 f^2 g t^2+3 f^2 g+2 f^2 t+2 f g^2+3 f g t+2 f t^2+f+3 g+t).
\end{align*}
The first and third factors vanish simultaneously only at $(f,g)=(\tilde{f}_{2,1},\tilde{g}_{2,1})$.

We finish this section, briefly discussing the case of Okamoto rationals.
At centers of the Weyl chambers in parameter space, i.e. $\alpha_1-\tfrac{1}{3}, \alpha_2-\tfrac{1}{3}\in\mathbb{Z}$, the fourth Painlev\'e equation over $\mathbb{Q}$ admits unique rational solutions known as Okamoto rationals, see \cites{noumibook,okamoto86}.
They can be generated by the translations $T_1$ and $T_2$, starting from the seed solution
\begin{equation}\label{eq:seedokamoto}
    f=-\tfrac{1}{3}t,\quad g=\tfrac{1}{3}t,\quad \alpha_1=\alpha_2=\tfrac{1}{3},
\end{equation}
yielding the family of Okamoto rationals
\begin{equation}\label{eq:famokamoto}
(\alpha_1,\alpha_2,f,g)=(m+\tfrac{1}{3},n+\tfrac{1}{3},f_{m,n}^{\text{Ok}},g_{m,n}^{\text{Ok}})\qquad (m,n\in\mathbb{Z}).
\end{equation}
Fixing any prime $p\neq 3$, the family has a good reduction modulo $p$, since the seed solution does so and the parameters do not lie on edges of Weyl chambers.



Like in the case of Hermite rationals, the family of Okamoto rationals also seem to be periodic in the indices with period $p$ after reducing it modulo $p$, 
\begin{equation*}
    (\tilde{f}_{(m+p,n)}^{\text{Ok}},\tilde{g}_{(m+p,n)}^{\text{Ok}})=(\tilde{f}_{(m,n)}^{\text{Ok}},\tilde{g}_{(m,n)}^{\text{Ok}}),\quad (\tilde{f}_{(m,n+p)}^{\text{Ok}},\tilde{g}_{(m,n+p)}^{\text{Ok}})=(\tilde{f}_{(m,n)}^{\text{Ok}},\tilde{g}_{(m,n)}^{\text{Ok}}).
\end{equation*}
We checked this periodicity in Mathematica for $p=2,5,7,11$ and $p=37$, but will not attempt to prove it here. In Table \ref{tab:okamotop5} reductions of some of the Okamoto rationals modulo $5$ are given\change{, where we note that $-1/3=3$ and $1/3=2$ in $\mathbb{F}_5$, so that the seed solution becomes
\begin{equation*}
    f=3t,\quad g=2t,\quad \alpha_1=\alpha_2=2,
\end{equation*}
in characteristic $5$.}

\begin{table}[H]
    \centering
    \begin{tabular}{|c|c|c|}\hline
     $m$\textbackslash $n$    & $0$ & $1$  \\\hline
    $0$  & $\left(3t,2t\right)$ & $\left(\frac{3 t (t^2+3)}{t^2+2},\frac{2 (t^2+2 t+4) (t^2+3 t+4)}{t (t^2+2)}\right)$ \\\hline
    $1$  & $\left(\frac{3 (t^2+t+1) (t^2+4 t+1)}{t (t^2+3)},\frac{2 t (t^2+2)}{t^2+3}\right)$ &  $\left(\frac{3 t (t^2+2)}{t^2+3},\frac{2 t (t^2+3)}{t^2+2}\right)$ \\\hline
    $2$  & $\left(\frac{3 (t^2+3) (t^2+2 t+4) (t^2+3 t+4)}{t (t^2+t+1) (t^2+4 t+1)},\frac{2 (t^2+t+1) (t^2+4 t+1)}{t (t^2+3)}\right)$ &  $\left(\frac{3 (t^2+2)}{t},2t\right)$ \\\hline
    $3$  & $\left(\frac{3 t^7}{(t^2+2) (t^2+2 t+4) (t^2+3 t+4)},\frac{2 (t^2+2 t+4) (t^2+3 t+4)}{t (t^2+2)}\right)$ &  $\left(\frac{3 t^5}{(t^2+2) (t^2+3)},\frac{2 t (t^2+2)}{t^2+3}\right)$ \\\hline
    $4$  & $\left(\frac{3 (t^2+2)}{t},\frac{2 t (t^2+3)}{t^2+2}\right)$ &  $\left(\frac{3 (t^2+3)}{t},\frac{2 (t^2+t+1) (t^2+4 t+1)}{t (t^2+3)}\right)$ \\\hline
    \end{tabular}
    \caption{Formulas for a selection of the members $(\tilde{f}_{(m,n)}^{\text{Ok}},\tilde{g}_{(m,n)}^{\text{Ok}})$ of the family \eqref{eq:famokamoto} reduced modulo $p=5$\change{, where we wrote all the} \change{rational functions in factored form over $\mathbb{F}_5$.}}
   \label{tab:okamotop5}
\end{table}

The corresponding values of the integrals of motions $\mathcal{I}_p$, for some primes $p$, are given by
\begin{align*}
    \mathcal{I}_2&=(1+t)^6, & \mathcal{I}_5&=3t^{15}, &
    \mathcal{I}_7&=6t^{21}, &
    \mathcal{I}_{11}&=9t^{33},\\
   \mathcal{I}_{13}&=t^{39}, & \mathcal{I}_{17}&=12t^{51}, &
    \mathcal{I}_{19}&=12t^{57}, &
    \mathcal{I}_{23}&=6t^{69}.
\end{align*}

\subsection{Singularities on fibres and sporadic algebraic solutions}\label{subsec:singularities}
Ohyama \cite{ohyama2010} showed that $q\Pone$ admits algebraic solutions when $q$ is a root of unity. In \cite{JR2026} it was observed that these algebraic solutions correspond precisely to singularities on fibres of integrals of motion constructed in that paper.

We find that a similar phenomenon takes place in the current context. Namely, singularities on fibres define algebraic solutions of $\Pfour$ and \change{$\textrm{d}\Ptwo$}.

Take for example the integral of motion $\mathcal{I}_2$ in characteristic $2$. Imposing that the curve $\mathcal{I}_2=c$ has a singularity, i.e. that the gradient of $\mathcal{I}_2$ with respect to $(f,g)$ vanishes at a point on the curve, leads to the system of three equations,
\begin{align}
    &\mathcal{I}_2=c,\label{eq:alg0}\\
    &tf^2+t^2 f+f+\alpha_1 t=0,\label{eq:alg1}\\
    &t\,g^2+t^2\, g+g+\alpha_2 t=0.\label{eq:alg2}
\end{align}
In characteristic $2$, $\Pfour$ decouples into two Riccati equations,
\begin{equation}\label{eq:piv_char2}
    D_t f=f^2+t f+\alpha_1,\quad D_t g=g^2+t g+\alpha_2,
\end{equation}
and a direct calculation shows that \eqref{eq:alg1} defines an algebraic solution to the first ODE and \eqref{eq:alg2} defines an algebraic solution to the second ODE. Remarkably, however,  \change{$\textrm{d}\Ptwo$} does not seem to decouple correspondingly.

Simplifying \eqref{eq:alg0} using \eqref{eq:alg1} and \eqref{eq:alg2}, gives the identity
\begin{equation*}
    (tf+\alpha_1)(t\,g+\alpha_2)+\frac{c}{t^2+1}=0.
\end{equation*}
Taking repeated resultants of the polynomials on the left-hand sides of the above equation and equations \eqref{eq:alg1} and \eqref{eq:alg2}, with respect to $f$ and $g$, then yields the following quartic constraint on $c$,
\begin{gather}
P=0,\nonumber\\
    P:=c^4+(1+t^2)(1+t^4)c^3+(\alpha_1+\alpha_1^2+\alpha_2+\alpha_2^2)(1+t^8)c^2+\label{eq:cconstraint}\\
    \alpha_1 \alpha_2(1+\alpha_1)(1+\alpha_2)(1+t^2)(1+t^8)c+\alpha_1^2\alpha_2^2(1+\alpha_1^2)(1+\alpha_2^2)(1+t^8),\nonumber
\end{gather}
as well as the following explicit formulas for $f$ and $g$ in terms of $c$,
\begin{equation}\label{eq:algsol}
\begin{aligned}
    f&=\frac{c^2+\alpha _1 \left(t^4+1\right) \left(c+\alpha _2 \left(\alpha _2+1\right) \left(\alpha _1+t^2\right)\right)}{t (t+1)^4 \left(c+\alpha _2 \left(\alpha _2+1\right) \left(t^2+1\right)\right)},\\
    g&=\frac{c^2 \left(\alpha _2+t^2+1\right)+\alpha _2 \left(\alpha _2+1\right) \left(t^4+1\right) \left(\alpha _1 \left(\alpha _1+1\right) \alpha _2+c\right)}{t \left(c^2+\alpha _1 \left(\alpha _1+1\right) \alpha _2^2 \left(\alpha _2+1\right) \left(t^4+1\right)\right)}.
\end{aligned}
\end{equation}
So $f$ and $g$ lie in the algebraic field extension of $\mathbb{F}_2(t,\alpha_1,\alpha_2)$ defined by the irreducible polynomial $P$ in \eqref{eq:cconstraint}.

Taking the formal derivative of \eqref{eq:cconstraint} with respect to $t$ gives
\begin{equation*}
    (c^2+(1+t^4)\alpha_1\alpha_2(1+\alpha_1)(1+\alpha_2))D_t c=0.
\end{equation*}
The quadratic factor on the left and the quartic equation \eqref{eq:cconstraint} have no common root over $\mathbb{F}_2(t,\alpha_1,\alpha_2)$ and the derivative $D_t$ thus extends uniquely to $\mathbb{F}_2(t,\alpha_1,\alpha_2)[c]/(P)$ with $D_t c=0$. Direct substitution of \eqref{eq:algsol} into \eqref{eq:piv_char2} re-affirms that it defines an algebraic solution to $\Pfour$.

Similarly, one can check that the quartic constraint \eqref{eq:cconstraint} is invariant under $\alpha_1\mapsto \alpha_1+1$ and $\alpha_2\mapsto \alpha_2+1$, so that we may choose $c$ invariant under these shifts, as expected. Then \eqref{eq:algsol} also defines an algebraic solution of \change{$\textrm{d}\Ptwo$}. 

Applying the same procedure in characteristic $p=3$, we impose that the curve $\mathcal{I}_3=c$ has a singularity, which leads to a system of three equations, that can be written as
\begin{align}
    \mathcal{I}_3&=c,\nonumber\\
    f&=\frac{t+g}{t^2+\alpha_2-\alpha_1-1}+t-g+\frac{\alpha_2}{g},\nonumber\\
    g&=\frac{t-f}{t^2+\alpha_2-\alpha_1+1}-t-f-\frac{\alpha_1}{f}.\label{eq:char3g}
\end{align}
Substituting the formula for $g$ into the formula for $f$ and vice versa leads to the following algebraic equations for $f$ and $g$ over $\mathbb{F}_3(t,\alpha_1,\alpha_2)$,
\begin{align*}
    & t \left(t^2+\alpha_2-\alpha_1-1\right) \left(t^2+\alpha_2-\alpha_1\right)f^3+ \left(t^6+\alpha_1-\alpha_2+\alpha_2^3-\alpha_1^3\right)f^2\\
    &+\alpha_1^2 \left(t^2+\alpha_2-\alpha_1+1\right)^2=0,\\
    & t \left(t^2+\alpha_2-\alpha_1+1\right) \left(t^2+\alpha_2-\alpha_1\right)g^3- \left(t^6+\alpha_1-\alpha_2+\alpha_2^3-\alpha_1^3\right)g^2\\
    &+\alpha_2^2 \left(t^2+\alpha_2-\alpha_1-1\right)^2=0.
\end{align*}
The first of these equations, can be taken together with the formula for $g$ in equation \eqref{eq:char3g}, to define a simultaneous algebraic solution to $\Pfour$ and \change{$\textrm{d}\Ptwo$} that lies in a degree three field extension of $\mathbb{F}_3(t,\alpha_1,\alpha_2)$. Correspondingly, $c$ also satisfies a degree three polynomial equation over $\mathbb{F}_3(t,\alpha_1,\alpha_2)$.

The same procedure can in principle be applied in any characteristics $p$, but explicit formulas become unwieldy already for relatively small primes $p$. It seems unlikely that the resulting algebraic solutions are reductions modulo $p$ of certain solutions of $\Pfour$ in characteristic $0$, like the Hermite rationals and Okamoto rationals discussed in the previous section. We therefore call these new algebraic solutions sporadic.

\section{Projective reduction of integrals to \change{\lowercase{d}$\Pone$}}\label{sec:projred}
It is well-known that, on special hyperplanes in parameter space, Painlev\'e IV admits special translations with step-size a half, whose square equals $T_1,T_2$ or another composition of those two translations and their inverses. Reductions of the dynamics to such hyperplanes are known as projective reductions \cite{KNT11}. In the case under consideration, the dynamics defined by these special translations are governed by the so called discrete first Painlev\'e equation, or $\change{\textrm{d}\Pone}$. We note that the discovery of $\change{\textrm{d}\Pone}$ predates the notion of projective reductions  \cite{fokasitskitaev91}.

In this section, we show that the integrals of motion of $\Pfour$ and \change{$\textrm{d}\Ptwo$} reduce to integrals of motion for $\change{\textrm{d}\Pone}$ under a projective reduction.

\subsection{A projective reduction to \NoCaseChange{\change{$\textrm{d}\Pone$}}}
Consider the element $s_2T_2$ of the Weyl group, which acts as
\begin{equation*}
    s_2T_2:\quad (\alpha_1,\alpha_2,f,g)\mapsto \left(\alpha_1+\alpha_2,1-\alpha_2,-g+\frac{\alpha_1}{f},f-g+t+\frac{\alpha_1}{f}\right).
\end{equation*}
It is a square root of the translation $T_1$ in $\widetilde W(A_2^{(1)})$, that is,
\begin{equation*}
    (s_2T_2)^2=T_1.
\end{equation*}
Such an element is also known as a quasi-translation \cites{shivariations19,shitranslations25}.

From here on, we work in odd characteristic $p$, so that $2$ is invertible. On the hyperplane $\alpha_2=\tfrac{1}{2}$ in parameter space, $s_2T_2$ defines a translation of $\alpha_1$ by $\tfrac{1}{2}$. Writing $\alpha_1=\tfrac{1}{2}\beta$, it becomes
\begin{equation*}
    \change{T_{\textrm{d}\Pone}:}\quad(\beta,f,g)\mapsto \left(\beta+1,-g+\frac{\beta}{2f},f-g+t+\frac{\beta}{2f}\right).
\end{equation*}
Writing $\change{T_{\textrm{d}\Pone}}(h)=\hat{h}$ and $\change{T_{\textrm{d}\Pone}}^{-1}(h)=\underaccent{\check}{h}$, we can write the dynamical system induced by $\change{T_{\textrm{d}\Pone}}$ as
\begin{equation*}
    \hat{f}=-g+\frac{\beta}{2f},\qquad \hat{g}=f-g+t+\frac{\beta}{2f}.
\end{equation*}
This yields $g=f+\underaccent{\check}{f}+t$ and
\begin{equation*}
    \hat{f}+f+\underaccent{\check}{f}=\frac{\beta}{2f}-t,
\end{equation*}
which is also known as discrete Painlev\'e I or $\change{\textrm{d}\Pone}$ \cite{fokasitskitaev91}.

\subsection{Reduction of integrals}
We now consider reductions of the integrals of motion to the hyperplane $\alpha_2=\tfrac{1}{2}$,
\begin{equation*}
    \mathcal{J}_{p}:=\mathcal{I}_{p}|_{(\alpha_1,\alpha_2)=(\frac{1}{2}\beta,\frac{1}{2})},
\end{equation*}
for odd primes $p$.
Then Theorem \ref{thm:invariance} and Theorem \ref{thm:sym} imply
\begin{equation*}
    \change{T_{\textrm{d}\Pone}}(\mathcal{J}_{p})=\mathcal{J}_{p}-t^p(\alpha_2^p-\alpha_2)=\mathcal{J}_{p},
\end{equation*}
where the last equality follows from the fact that $\alpha_2^p=\alpha_2$ when $\alpha_2=\tfrac{1}{2}$ by Fermat's little theorem.
So, we see that the integral of motion $\mathcal{I}_p$ of \change{$\textrm{d}\Ptwo$} reduces to an integral of motion $\mathcal{J}_{p}$ of $\change{\textrm{d}\Pone}$.

When $p=3$, we have the following explicit formula for the integral,
\begin{align*}
    \mathcal{J}_{3}=&\mathcal{J}_{3}^{(1)}\mathcal{J}_{3}^{(2)},\\
    \mathcal{J}_{3}^{(1)}=&\,t^2 \beta +t^2 \underaccent{\check}{f} f+t^2+2 t \beta  \underaccent{\check}{f}+2 t \beta  f+2 t \underaccent{\check}{f}^2 f+2 t \underaccent{\check}{f} f^2+t \underaccent{\check}{f}+\beta  \underaccent{\check}{f}^2+2 \beta  \underaccent{\check}{f} f+\beta  f^2\\
    &+\underaccent{\check}{f}^3 f+2 \underaccent{\check}{f}^2 f^2+\underaccent{\check}{f} f^3+2 \underaccent{\check}{f} f+2 f^2+2,\\
    \mathcal{J}_{3}^{(2)}=&\,t \beta^2+2 t \beta  \underaccent{\check}{f} f+2 t \beta +t \underaccent{\check}{f}^2 f^2+2 t \underaccent{\check}{f} f+\beta^2 \underaccent{\check}{f}+\beta^2 f+2 \beta  \underaccent{\check}{f}^2 f+2 \beta  \underaccent{\check}{f} f^2+\beta  f\\
    &+\underaccent{\check}{f}^3 f^2+\underaccent{\check}{f}^2 f^3+\underaccent{\check}{f} f^2+f.
\end{align*}
Direct computation further shows that
\begin{equation*}
    \change{T_{\textrm{d}\Pone}}(\mathcal{J}_{3}^{(1)})=\frac{2}{f}\mathcal{J}_{3}^{(2)},\qquad
    \change{T_{\textrm{d}\Pone}}(\mathcal{J}_{3}^{(2)})=2f\mathcal{J}_{3}^{(1)}.
\end{equation*}

{\allowdisplaybreaks
Similarly, we obtain the following factorised form for $\mathcal{J}_{5}$,
\begin{align*}
    \mathcal{J}_{5}=&\mathcal{J}_{5}^{(1)}\mathcal{J}_{5}^{(2)},\\
    \mathcal{J}_{5}^{(1)}=&\,2 t^3 \beta^2+2 t^3 \beta \underaccent{\check}{f} f+3 t^3 \beta+3 t^3 \underaccent{\check}{f}^2 f^2+4 t^3 \underaccent{\check}{f} f+t^3+t^2 \beta^2 \underaccent{\check}{f}+t^2 \beta^2 f+t^2 \beta \underaccent{\check}{f}^2 f\\
    &+t^2 \beta \underaccent{\check}{f} f^2+t^2 \beta \underaccent{\check}{f}+2 t^2 \beta f+4 t^2 \underaccent{\check}{f}^3 f^2+4 t^2 \underaccent{\check}{f}^2 f^3+3 t^2 \underaccent{\check}{f}^2 f+t^2 \underaccent{\check}{f} f^2+t^2 \underaccent{\check}{f}\\
    &+3 t^2 f+t \beta^2 \underaccent{\check}{f}^2+2 t \beta^2 \underaccent{\check}{f} f+t \beta^2 f^2+t \beta \underaccent{\check}{f}^3 f+2 t \beta \underaccent{\check}{f}^2 f^2+3 t \beta \underaccent{\check}{f}^2+t \beta \underaccent{\check}{f} f^3\\
    &+3 t \beta \underaccent{\check}{f} f+t \beta+4 t \underaccent{\check}{f}^4 f^2+3 t \underaccent{\check}{f}^3 f^3+4 t \underaccent{\check}{f}^3 f+4 t \underaccent{\check}{f}^2 f^4+4 t \underaccent{\check}{f}^2 f^2+3 t \underaccent{\check}{f} f+4 t f^2\\
    &+2 t+2 \beta^2 \underaccent{\check}{f}^3+\beta^2 \underaccent{\check}{f}^2 f+\beta^2 \underaccent{\check}{f} f^2+2 \beta^2 f^3+2 \beta \underaccent{\check}{f}^4 f+\beta \underaccent{\check}{f}^3 f^2+\beta \underaccent{\check}{f}^2 f^3+\beta \underaccent{\check}{f}^2 f\\
    &+2 \beta \underaccent{\check}{f} f^4+2 \beta \underaccent{\check}{f} f^2+\beta \underaccent{\check}{f}+\beta f^3+\beta f+3 \underaccent{\check}{f}^5 f^2+4 \underaccent{\check}{f}^4 f^3+4 \underaccent{\check}{f}^3 f^4+3 \underaccent{\check}{f}^3 f^2\\
    &+3 \underaccent{\check}{f}^2 f^5+\underaccent{\check}{f}^2 f^3+\underaccent{\check}{f}^2 f+3 \underaccent{\check}{f} f^4+3 \underaccent{\check}{f} f^2+4 \underaccent{\check}{f}+2 f^3,\\
    \mathcal{J}_{5}^{(2)}=&\,t^2 \beta^3+4 t^2 \beta^2 \underaccent{\check}{f} f+t^2 \beta^2+2 t^2 \beta \underaccent{\check}{f}^2 f^2+t^2 \beta \underaccent{\check}{f} f+3 t^2 \beta+2 t^2 \underaccent{\check}{f}^3 f^3+4 t^2 \underaccent{\check}{f}^2 f^2\\
    &+4 t^2 \underaccent{\check}{f} f+2 t \beta^3 \underaccent{\check}{f}+2 t \beta^3 f+3 t \beta^2 \underaccent{\check}{f}^2 f+3 t \beta^2 \underaccent{\check}{f} f^2+t \beta^2 \underaccent{\check}{f}+3 t \beta^2 f+4 t \beta \underaccent{\check}{f}^3 f^2\\
    &+4 t \beta \underaccent{\check}{f}^2 f^3+t \beta \underaccent{\check}{f}^2 f+3 t \beta \underaccent{\check}{f} f^2+3 t \beta f+4 t \underaccent{\check}{f}^4 f^3+4 t \underaccent{\check}{f}^3 f^4+4 t \underaccent{\check}{f}^3 f^2+2 t \underaccent{\check}{f}^2 f^3\\
    &+4 t \underaccent{\check}{f} f^2+2 t f+\beta^3 \underaccent{\check}{f}^2+2 \beta^3 \underaccent{\check}{f} f+\beta^3 f^2+4 \beta^2 \underaccent{\check}{f}^3 f+3 \beta^2 \underaccent{\check}{f}^2 f^2+4 \beta^2 \underaccent{\check}{f} f^3\\
    &+2 \beta^2 \underaccent{\check}{f} f+2 \beta^2 f^2+2 \beta^2+2 \beta \underaccent{\check}{f}^4 f^2+4 \beta \underaccent{\check}{f}^3 f^3+2 \beta \underaccent{\check}{f}^2 f^4+2 \beta \underaccent{\check}{f}^2 f^2+2 \beta \underaccent{\check}{f} f^3\\
    &+3 \beta \underaccent{\check}{f} f+3 \beta f^2+3 \beta+2 \underaccent{\check}{f}^5 f^3+4 \underaccent{\check}{f}^4 f^4+2 \underaccent{\check}{f}^3 f^5+3 \underaccent{\check}{f}^3 f^3+3 \underaccent{\check}{f}^2 f^4+\underaccent{\check}{f}^2 f^2\\
    &+4 \underaccent{\check}{f} f^3+4 \underaccent{\check}{f} f+4 f^2,
\end{align*}}
with
\begin{equation*}
    \change{T_{\textrm{d}\Pone}}(\mathcal{J}_{5}^{(1)})=\frac{1}{f}\mathcal{J}_{5}^{(2)},\qquad
    \change{T_{\textrm{d}\Pone}}(\mathcal{J}_{5}^{(2)})=f\mathcal{J}_{5}^{(1)}.
\end{equation*}
One may expect factorisations as above to occur in general for odd primes $p$. We leave this line of enquiry for future research.

\section{Conclusion}\label{sec:conclusion}
In \cite{JR2026}, a method was introduced to obtain integrals of motion for $q$-difference Painlev\'e equations from their Lax pairs when $q$ has finite multiplicative order. In this paper, we presented an analogous method for difference and differential Painlev\'e equations in finite characteristic using difference Lax pairs, illustrated for \change{$\textrm{d}\Ptwo$} and $\Pfour$. In particular, our method applies in the context of finite fields and may be useful to study the arithmetic properties and cryptographic potential of Painlev\'e equations.

We found that reducible fibres and singularities on fibres correspond respectively to reductions to Riccatti equations of $\Pfour$ for special parameter values and algebraic solutions of \change{$\textrm{d}\Ptwo$} and $\Pfour$. We further showed that the integrals of motions reduce to integrals for $\change{\textrm{d}\Pone}$ under a projective reduction.

Based on the findings for $q\Pone$ in \cite{JR2026}, it is natural to ask whether
 the fibres of the integrals of motion are generally genus one curves also in the current setting. Other avenues for future exploration include the special multiplicative structure in the integrals of motion of $\change{\textrm{d}\Pone}$ observed in Section \ref{sec:projred}, as well as application of the method to other Painlev\'e equations and members higher up in Painlev\'e hierarchies.

\appendix

\section{Explicit formulas for integrals of motion}\label{appendix:explicitformulas}

Here are the first few integrals of motion, as defined in Definition \ref{def:integralofmotion}.

\begin{align*}
    \mathcal{I}_2=&+f^4 g^2+f^2 g^4+t^2 f^2 g^2 +
    t f^2 g+t f g^2 +t^2 f g 
    +\alpha_2^2 f^2+f g\\
    &+\alpha_2 t f +\alpha_1^2 g^2+\alpha_1 t g 
    +\alpha_1 \alpha_2,\\
    \mathcal{I}_3=&-f^6 g^3+f^3 g^6-t^3 f^3 g^3
    +t^3 f g+t^2 f^2 g-t^2 f g^2+\alpha_2 t^2 f+\alpha_1 t^2 g\\
    &-\alpha_1 t f g+\alpha_2 t f g-\alpha_2^3 f^3-\alpha_1 f^2 g+\alpha_2 f^2 g-f^2 g+\alpha_1 f g^2-\alpha_2 f g^2- f g^2\\
    &- \alpha_1^3 g^3+\alpha_1\alpha_2  t-\alpha_1\alpha_2  f+\alpha_2^2 f-\alpha_2 f-\alpha_1^2 g+\alpha_1 g+\alpha_1\alpha_2  g,\\
    \mathcal{I}_5=&-f^{10} g^5+f^5 g^{10}-t^5 f^5 g^5+2 t  f^4 g^2+t  f^3 g^3- t ^2 f^3 g^2+2 t  f^2 g^4+t ^2 f^2 g^3\\
    &+2 t ^3 f^2 g^2+t ^4 f^2 g- t ^4 f g^2+t ^5 f g- \alpha _2^5 f^5+3 f^3 g^2- \alpha _2 t  f^3 g+3 f^2 g^3\\
    &- \alpha _1 t  f^2 g^2+\alpha _2 t  f^2 g^2+2 \alpha _1 t ^2 f^2 g+2 \alpha _2 t ^2 f^2 g- t ^2 f^2 g+\alpha _2 t ^4 f+\alpha _1 t  f g^3\\
    &+2 \alpha _1 t ^2 f g^2+2 \alpha _2 t ^2 f g^2- t ^2 f g^2+2 \alpha _1 t ^3 f g+3 \alpha _2 t ^3 f g- \alpha _1^5 g^5+\alpha _1 t ^4 g\\
    &+\alpha _1 \alpha _2 t ^3+2 \alpha _2^2 t  f^2+\alpha _1^2 f^2 g+\alpha _2^2 f^2 g- \alpha _1 f^2 g- \alpha _2 f^2 g+\alpha _1 \alpha _2 f^2 g+3 f^2 g\\
    &+3 \alpha _2^2 t ^2 f- \alpha _2 t ^2 f+2 \alpha _1 \alpha _2 t ^2 f- \alpha _1^2 f g^2- \alpha _2^2 f g^2+\alpha _1 f g^2+\alpha _2 f g^2\\
    &- \alpha _1 \alpha _2 f g^2+2 f g^2+\alpha _1^2 t  f g+\alpha _2^2 t  f g+2 \alpha _1^2 t  g^2+2 \alpha _1^2 t ^2 g+\alpha _1 t ^2 g\\
    &+3 \alpha _1 \alpha _2 t ^2 g
    +3 \alpha _1^2 \alpha _2  t +2  \alpha _1 \alpha _2^2 t +\alpha _2 \alpha _1^2 f+\alpha _2^2 \alpha _1 f-  \alpha _1 \alpha _2 f+\alpha _2^3 f+\alpha _2^2 f\\
    &+3 \alpha _2 f+\alpha _1^3 g+\alpha _1^2 g+\alpha _2 \alpha _1^2 g+ \alpha _1 \alpha _2^2 g+3 \alpha _1 g-  \alpha _1 \alpha _2 g.
\end{align*}

\section{Explicit formulas for coefficient matrices in Lax form}\label{sec:laxcoef}
We have the following explicit formulas for the coefficient matrices appearing in Equations \eqref{eq:piv_laxB} and \eqref{eq:piv_laxC}.
The matrix $B^{(2)}(z)$ is a $2\times 2$ matrix polynomial of degree two in $z$,
\begin{equation*}
    B^{(2)}(z)=\begin{bmatrix}
        1 & 0\\
        w^{-1}(\beta_1-t\,z) & w^{-1}
    \end{bmatrix}
    \begin{bmatrix}
        \beta_2+t\,z & w\\
        \alpha_2+z & 0
    \end{bmatrix},
\end{equation*}
and the inverse of $B^{(1)}(z)$ is a $2\times 2$ matrix polynomial of degree two in $z$,
\begin{equation*}
    B^{(1)}(z)^{-1}=\begin{bmatrix}
        1 & 0\\
        w^{-1}(\beta_3-t\, z) & w^{-1}
    \end{bmatrix}
    \begin{bmatrix}
        \beta_4+t z &  w\\
        z-\alpha_1-1 & 0
    \end{bmatrix},
\end{equation*}
where the $\beta_k$, $1\leq k\leq 4$, are given by
\begin{align*}
    \beta_1&=t(1+f g-\alpha_1)+g(t^2-1)\frac{f g-\alpha_1}{f g+\alpha_2}-t\,g^2  \frac{(f g-\alpha_1)^2}{(f g+\alpha_2)^2},\\
    \beta_2&=t\,f g-g \frac{f g-\alpha_1}{f g+\alpha_2},\\
    \beta_3&=t-t\,f g+f \frac{f g+\alpha_2}{f g-\alpha_1},\\
    \beta_4&=-t(f g+\alpha_2)+(1+t^2)f \frac{f g+\alpha_2}{f g-\alpha_1}+t f^2 \frac{(f g+\alpha_2)^2}{(f g-\alpha_1)^2}.
\end{align*}
The matrix $C(z)$ is a $2\times 2$ matrix polynomial of degree two in $z$,
\begin{equation*}
 C(z)=C_0+z \begin{bmatrix}
     t & 0\\
     \frac{t^2(1-2 f g)}{w} & -t
 \end{bmatrix}+z^2\begin{bmatrix}
     0 & 0\\
     -\frac{t^2}{w} & 0
 \end{bmatrix},
\end{equation*}
with
\begin{equation*}
    C_0=\begin{bmatrix}
        t f g-\frac{t}{2} & w\\
         w^{-1}[-fg(t^2 f g+tf+t\,g-t^2-1)+\alpha_1(t\,g-1)-\alpha_2(tf-1)]    & \frac{t}{2}-t f g
    \end{bmatrix}.
\end{equation*}

\section{Proof of a technical lemma}

The goal of this section is to prove Lemma \ref{lem:traceconstant}. It requires two preparatory lemmas. The content of the first is a cyclic sum identity for polynomials over finite fields.
\begin{lemma}\label{lem:polevalsum}
     Given a prime $p$ and a polynomial
     \begin{equation*}
         g(z)=g_0+g_1 z+g_2 z^2+\ldots+g_{d(p-1)}z^{d(p-1)}\in\mathbb{F}_p[z],
     \end{equation*}
     for some $d\geq 0$, the cyclic sum of evaluations of the polynomial at all elements of $\mathbb{F}_p$ equals
     \begin{equation*}
        \sum_{n=0}^{p-1}g(n)=-(g_{p-1}+g_{2p-2}+\ldots+g_{d(p-1)}).
     \end{equation*}
\end{lemma}
\begin{proof}
    Note that
    \begin{equation}\label{eq:gsum}
       \sum_{n=0}^{p-1}g(n)=\sum_{k=0}^{d(p-1)} g_k \sum_{n=0}^{p-1}n^k.
    \end{equation}
    We now recall the well-known power sum congruence, see e.g. \cite{macmillansondow},
        \begin{equation*}
        \sum_{n=0}^{p-1}n^k\equiv \begin{cases}
            -1 &\text{if $p-1\mid k$},\\
            0 &\text{if $p-1 \nmid k$},\\
        \end{cases}
    \end{equation*}
            modulo $p$. The lemma follows by applying this identity to each of the inner summations on the right-hand side of equation \eqref{eq:gsum}.
\end{proof}

\begin{lemma}\label{lem:polsum}
     Given a prime $p$ and a polynomial of degree at most $2p-2$,
     \begin{equation*}
         g(z)=g_0+g_1 z+g_2 z^2+\ldots+g_{2(p-1)}z^{2p-2}\in\mathbb{F}_p[z],
     \end{equation*}
      the following cyclic sum equals a constant,
     \begin{equation*}
        \sum_{n=0}^{p-1}g(z+n)=-(g_{p-1}+g_{2p-2}).
     \end{equation*}
\end{lemma}
\begin{proof}
Denote
\begin{equation*}
    G(z)=\sum_{n=0}^{p-1}g(z+n),
\end{equation*}
then $G(z)$ is a polynomial of degree at most $2p-2$ that satisfies $G(z+1)=G(z)$. It follows from these two facts that
\begin{equation*}
    G(z)=G_0+G_1 z+G_p z^p,
\end{equation*}
for some $G_0,G_1,G_p\in\mathbb{F}_p$ with $G_1+G_p=0$. We can compute $G_0$ by applying Lemma \ref{lem:polevalsum} to $g(z)$, yielding
\begin{equation*}
    G_0=-(g_{p-1}+g_{2p-2}).
\end{equation*}
Next, to compute $G_1$, we apply Lemma \ref{lem:polevalsum} to $g'(z)$, yielding
\begin{equation*}
    0=\sum_{n=0}^{p-1}g'(z+n)=G'(z)=G_1,
\end{equation*}
so that $G_1=0$. Finally, since $G_1+G_p=0$, we have $G_p=0$ and the lemma follows.
\end{proof}

\begin{proof}[Proof of Lemma \ref{lem:traceconstant}]\label{sec:appendix_proof}
The case $p=2$ is done by direct calculation, so we assume that $p$ is an odd prime.
We prove the statement of the lemma by exploiting the cycling invariance of traces of products of matrices in combination with Lemma \ref{lem:polsum}. For this purpose we introduce the notation $\mathbf{j}=(j_0,\ldots,j_{p-1})\in \{0,1,2\}^p$ for length $p$ tuples of indices and define, for $0\leq d\leq 2p$, the set of indices that sum up to $d$,
  \begin{equation*}
      J_d=\{\mathbf{j}\in \{0,1,2\}^p:j_0+\ldots+j_{p-1}=d\}.
  \end{equation*}
We then have a cyclic action on $J_d$ generated by the permutation
\begin{equation*}
    \sigma: J_d\rightarrow J_d, \mathbf{j}\mapsto (j_1,\ldots,j_{p-1},j_0).
\end{equation*}
Since $p$ is prime, any element under this action is either a fixed point or has an orbit of length $p$.
  
We proceed to decompose the main trace of interest as follows,
    \begin{equation}\label{eq:trace_decomp}
        \operatorname{Tr}\left[B(z+p-1)\cdot B(z+p-2)\cdot\ldots\cdot B(z+1)\cdot B(z)\right]=\sum_{d=0}^{2p} H_d(z),
    \end{equation}
where
\begin{equation*}
   H_d(z)=\sum_{\mathbf{j}\in J_d}g_\mathbf{j}(z) \operatorname{Tr}\left[B_{j_{p-1}}B_{j_{p-2}}\cdot\ldots\cdot B_{j_1}B_{j_0}\right],
\end{equation*}
with
\begin{equation*}
    g_\mathbf{j}(z)=(z+p-1)^{j_{p-1}}(z+p-2)^{j_{p-2}}\cdot\ldots\cdot (z+1)^{j_1}(z+0)^{j_0}.
\end{equation*}

To compute $H_d(z)$, we distinguish five cases, $d=2p$, $d=2p-1$, $d=p$, $d=0$ and the remaining values of $d$, that is, $1\leq d\leq 2p-2$ with $d\neq p$.

First, we consider $d=2p$, in which case $J_d=\{(2,2,\ldots,2)\}$ consists of just one element and we have
\begin{align}
    H_{2p}(z)&=(z+p-1)^{2}(z+p-2)^{2}\cdot\ldots\cdot (z+1)^{2}(z+0)^{2}\operatorname{Tr}B_2^p \nonumber\\
    &=(z^p-z)^2 \operatorname{Tr}B_2^p. \label{eq:H2p} 
\end{align}
Similarly, when $d=2p-1$, $J_d$ consists of one cycle of tuples under the action of $\langle \sigma \rangle$,
\begin{equation*}
    J_{2p-1}=\{\sigma^n(1,2,2,\ldots,2):0\leq n\leq p-1\},
\end{equation*}
and, by the cyclic invariance of traces of products of matrices,
\begin{align}
    H_{2p-1}(z)&=(z^p-z)^2\left(\frac{1}{z+p-1}+\frac{1}{z+p-2}+\ldots+\frac{1}{z}\right) \operatorname{Tr}B_2^{p-1}B_1\nonumber\\
    &=(z^p-z) \left(\frac{d}{dz}(z^p-z)\right)\operatorname{Tr}B_2^{p-1}B_1\nonumber\\
    &=-(z^p-z)\operatorname{Tr}B_2^{p-1}B_1. \label{eq:H2pm1}
\end{align}

Next, we consider the case $1\leq d\leq 2p-2$ with $d\neq p$. Then we know that the action of $\sigma$ on $J_d$ has no fixed points. Choosing a set of representatives of the corresponding quotient,
\begin{equation*}
    J_d/\langle \sigma \rangle =\{[\mathbf{j}^{(r)}]:1\leq r \leq R\},
\end{equation*}
where $R=|J_d|/p$, we have, by the cycling invariance of traces of products of matrices,
\begin{equation*}
   H_d(z)=\sum_{r=1}^R\sum_{n=0}^{p-1} g_{\mathbf{j}^{(r)}}(z+n) \operatorname{Tr}\left[B_{j_{p-1}}B_{j_{p-2}}\cdot\ldots\cdot B_{j_1}B_{j_0}\right].
\end{equation*}
By Lemma \ref{lem:polsum}, we know that each of the inner sums
\begin{equation*}
    \sum_{n=0}^{p-1} g_{\mathbf{j}^{(r)}}(z+n),
\end{equation*}
is a constant, so that $H_d(z)$ itself is just a constant, that is,
\begin{equation}\label{eq:Hgen}
    H_d(z)=H_d(0),
\end{equation}
for $1\leq d\leq 2p-2$ with $d\neq p$.

When $d=p$, we can repeat the same computation, but we need to consider the special fixed point $\mathbf{j}=(1,1,\dots 1)\in H_p$ separately. This tuple of indices yields the following contribution to $H_p(z)$,
\begin{equation*}
    (z+p-1)(z+p-2)\cdot\ldots\cdot (z+1)(z+0) \operatorname{Tr}B_1^p=(z^p-z)\operatorname{Tr}B_1^p.
\end{equation*}
In other words, we have
\begin{equation}\label{eq:Hp}
    H_p(z)=H_p(0)+(z^p-z)\operatorname{Tr}B_1^p.
\end{equation}

For the remaining case $d=0$, we have
\begin{equation*}
    H_0(z)=\operatorname{Tr}B_0^p,
\end{equation*}
 so that the statement of the lemma follows from combining this equation and equations \eqref{eq:trace_decomp}, \eqref{eq:H2p}, \eqref{eq:H2pm1}, \eqref{eq:Hgen} and  \eqref{eq:Hp}.
 \end{proof}

\end{document}